\documentclass{article}
\usepackage[utf8]{inputenc}
\usepackage{graphicx}
\usepackage[letterpaper,margin=1.2in]{geometry}
\usepackage{babel}
\usepackage{amsmath}
\usepackage{amsthm} 
\usepackage{amsfonts}
\usepackage{pict2e}
\usepackage{float}
\usepackage[dvipsnames]{xcolor}
\usepackage{algorithm}
\usepackage[noend]{algpseudocode}
\usepackage{hyperref}
\usepackage{cleveref}
\usepackage[skip = 0.2\baselineskip plus 1.5pt]{parskip}
\usepackage{thmtools, thm-restate}

\newcommand{\R}{\mathbb{R}}
\newcommand{\opt}{\textsc{opt}}
\newcommand{\eps}{\varepsilon}
\newcommand{\thresh}{\Theta}

\newcommand{\bc}{\mathfrak{c}}
\newcommand{\err}{\mathfrak{e}}

\newcommand{\parent}{\textsc{Parent}}

\DeclareMathOperator{\poly}{poly}
\DeclareMathOperator{\dist}{dist}
\DeclareMathOperator{\cost}{cost}
\DeclareMathOperator{\nval}{NoisyValue}
\DeclareMathOperator{\val}{Value}

\DeclareMathOperator{\argmin}{argmin}
\DeclareMathOperator{\level}{level}

\newcommand{\calA}{\mathcal{A}} 
\newcommand{\calC}{\mathcal{C}} 
\newcommand{\calE}{\mathcal{E}} 
\newcommand{\calG}{\mathcal{G}}
\newcommand{\calM}{\mathcal{M}}

\newcommand{\calS}{\mathcal{S}}
\newcommand{\calU}{\mathcal{U}}
\newcommand{\calX}{\mathcal{X}}

\newcommand{\algLowDim}{\textsc{Private-Clustering}}
\newcommand{\makePrivate}{\textsc{MakePrivate-HighDimension}}
\newcommand{\makePrivLowDim}{\textsc{MakePrivate}}
\newcommand{\hdkm}{\textsc{Private-$k$-Means}}
\newcommand{\mettuP}{\textsc{RecursiveGreedy}}

\newcommand{\sumEst}{\textsc{Sum}}
\newcommand{\sumNorm}{\textsc{SumNorm}}

\newtheorem{theorem}{Theorem}

\newtheorem{lemma}[theorem]{Lemma}

\newtheorem{fact}[theorem]{Fact}

\newtheorem{definition}[theorem]{Definition}

\newcommand{\lpar}{\left(}
\newcommand{\rpar}{\right)}
\newcommand{\lbra}{\left\{}
\newcommand{\rbra}{\right\}}
\newcommand{\logor}{\left\|}
\newcommand{\rnor}{\right\|}

\newcounter{sideremark}

 \newcommand{\david}[1]{
 }

\title{Differential Privacy for Clustering Under Continual Observation}
\author{Max Dupré la Tour\thanks{McGill University}, Monika Henzinger\thanks{IST Austria}, David Saulpic\footnotemark[2]}
\date{}

\begin{document}
\maketitle

\begin{abstract}
    We consider the problem of clustering  privately a dataset in $\R^d$ that undergoes insertion and deletion of points. Specifically, we give an $\eps$-differentially private clustering mechanism for the $k$-means objective under continual observation. This is the first approximation algorithm for that problem with an additive error that depends only logarithmically in the number $T$ of updates. The multiplicative error is almost the same as non privately.
    To do so we show how to perform dimension reduction under continual observation and combine it with a differentially private greedy approximation algorithm for $k$-means.
    \end{abstract}

\section{Introduction}
The massive, continuous and automatic collection of personal data of public as well as private organisations has raised privacy concerns, in legal terms  \cite{gdpr} but also in terms of citizens' demands \cite{nyt, nyt2}.
To answer those, formal privacy standards for algorithms were defined and developed: the most prominent one is Differential Privacy~\cite{dworkDef}. It allows to have a formal definition of \emph{privacy}, and therefore to give algorithms with provable privacy guarantees. 
Differentially private algorithms are now deployed quite massively. For instance, the U.S. Census Bureau uses it to release information from a private Census \cite{abowd2018us}, Apple collects data from its phone's users with an algorithm that respects differential privacy~\cite{apple1, apple2}, and Google develops an extensive library of private algorithms ready to be used~\cite{googleLib}.
Differential privacy ensures that those algorithms ``behave roughly the same way'' on two databases that differ by a single element. 
The idea behind it is that one cannot infer from the result whether a specific element is present in the database, i.e.~protecting against membership-inference attacks.

However, this guarantee is only valid for a \emph{static} database. As real-life data often evolves over time -- as in Apple's example, where personal data is collected and transferred every day -- a definition of privacy that accounts for such changes is necessary.
To address that issue, an extension of differential privacy was proposed by Dwork, Naor, Pitassi, and Rothblum~\cite{dwork2010differential}: an algorithm that respects privacy \emph{under continual observation} is given a sequence of updates to its data set, one per time step, and computes an output at each time step, such that differential privacy is preserved for all time steps.
A sequence of publications has studied  maintaining a mere counter under that model~\cite{dwork2010differential, ChanSS11, jain2021price,ghazi2022private, henzinger2023almost,HenzingerU22}, and other basic tools such as frequency estimation \cite{epasto2023differentially}, histograms \cite{cardoso2022differentially}, heavy hitters \cite{heavyHitters}, and various graph properties~\cite{FichtenbergerHO21}.

In this article, we consider a more complex problem, namely clustering. 
This is one of the most central problems in unsupervised machine learning, with many applications ranging from community detection to duplicate detection. 
One example of an application is to report clusters of people infected by a disease, based on the web search in geographic regions. The input data for that application is inherently evolving -- new web searches are continuously performed -- and highly sensitive. 
The question we ask is therefore:

\fbox{Is it possible to maintain a provably good clustering that is private under continual observation?}

\subsection{Our Results.}

In order to answer this question, 
we need to formalize what a good clustering means. For 
that, we will focus on one of the most prominent clustering objectives, namely the $k$-means problem. 
In this problem, the input data are points in a metric space, and the goal is to find $k$ points from the space (the \emph{centers}) in order to minimize the cost, namely the sum of the squared distance of each point to its closest center. We focus more precisely on the case where the data comes from a high-dimensional Euclidean space $\R^d$, arguably the main use case for $k$-means clustering. 

The $k$-means problem has been heavily studied and is very well understood in the non-private setting (see e.g. Cohen-Addad et al.~\cite{cohen2022improved} for the static case, Henzinger and Kale \cite{HenzingerK20} for dynamic, and references therein), as well as with (static) differential privacy constraint \cite{StemmerK18, kdd}. In the static differentially private setting, two types of results emerge: those with optimal \emph{multiplicative} error \cite{badih_approximation, ChangG0M21, dpHD} (and large additive error), and those with optimal \emph{additive} error (and large multiplicative one) \cite{anamayclustering}.

However, nothing is known about
maintaining a cluster for a sensitive dataset that is evolving over time.
We 
initiate the study of differentially private clustering under continual observation
and present the first approximation algorithm for $k$-means clustering that is differentially private under continual observation. More specifically, our algorithm has the same optimal multiplicative error, as \cite{badih_approximation} for the static case and an additive error that depends only polylogarithmically on the length  of the update sequence.

More formally, the setting is as follows: the algorithm is given a \emph{diameter} $\Lambda \in \R^+$, and a stream of input data of length $T$, i.e., at each time step, a point from the $d$-dimensional ball $B(0, \Lambda)$ (centered in the origin with radius $\Lambda$) is either added or removed from the dataset. After each update, our algorithm needs to output a solution that minimizes the $k$-means cost in the metric space $(\R^d, \ell_2)$. We denote the optimal cost at time $t$ by $\opt_t$, $T$ the total number of time steps and $n$ the maximal size of the dataset. For the definition of privacy that we use, namely \emph{$\epsilon$-
differential privacy ($\eps$-DP)}, we say that two input streams are \emph{neighboring} if they differ in at most one time step, i.e., in a single input point. We give more details in \Cref{sec:prelim}. 
We show the following result.

\begin{theorem}\label{thm:main}
For any $\alpha >0$, there exists an algorithm  that is $\eps$-differentially private under continual observation and computes, with probability at least $0.99$ for all time steps $t$ simultaneously, a solution  to $k$-means and its cost in $\R^d$ with cost at most 
$$(1+\alpha)w^* \cdot \opt_t + k^{O_\alpha(1)}  d^2 \log(n)^4 \cdot \log(T)^{3+0.001}  \cdot \Lambda^2/{\eps},$$
where $w^*$ is the best approximation ratio for non-private, static $k$-means clustering.
\end{theorem}

For the slightly weaker $(\eps, \delta)$-differential privacy, the additive error decreases to $k^{O(1)} d^2 \log(1/\delta) \cdot \log(n)^4 \log(T)^{2+0.001} \cdot \Lambda^2/\eps$.
Our result is actually even stronger: it computes a solution not only to $k$-means, but to all $k'$-means for $k' \leq k$, with the same multiplicative and additive error guarantee as above.
This allows, for instance, to use the standard elbow method\footnote{This is a heuristic to compute the ``right" value of $k$, as follows: the curve of the $k$-means cost as a function of $k$ has an inflection point, which is defined as ``right" value of $k$.} to select the ``best" value for $k$, without any further loss of privacy. 
Note that the algorithm does not need to know $n$ in advance, nor the total number of time steps.

To understand our additive error term, we note that, even in the static setting, any algorithm must suffer from an additive error $\Omega\lpar\frac{kd}{\eps} \cdot \Lambda^2\rpar$~\cite{GuptaLMRT10, anamayclustering}. Under continual observation any differentially-private $k$-means algorithm with multiplicative error $(1+\alpha)$ (for some  small $\alpha>0$) can be used to solve the differentially-private counting problem  and, thus, the lower bound of $\Omega(\log (T))$ for the additive error of continual counting implies an $\Omega(\log (T))$ lower bound for the additive error of $k$-means.

Our algorithm uses a differentially-private counting algorithm under continual observation in a black-box manner, which  contributes an $\Omega(\log (T)^2)$ term to the additive error. This is slightly increases by 
a $\log(T)^{1/2 + 0.001}$ factor in our algorithm. Note, that any improvement in the additive error of $\epsilon$-DP continual counting would immediately improve the additive error of our algorithm.

\paragraph*{Extension to $k$-median.}

Our techniques partially extend to $k$-median, for which the cost function is the sum of distances to the centers (instead of the distances squared). We show the following result:

\begin{restatable}{theorem}{kmed}
    \label{thm:kmed}
For any $\alpha >0$, there exists an algorithm  that is $\eps$-differentially private under continual observation and computes, for all time steps $t$ individually, a solution  to $k$-median and its cost in $\R^d$ with cost at most 
$$(2+\alpha)w^* \cdot \opt_t + k^{O_\alpha(1)}  d^3 \log(n)^4 \cdot \log(T)^{3/2}  \cdot \Lambda/{\eps},$$
 with probability at least $0.99$, where  $w^*$ is the best approximation ratio for non-private, static $k$-median clustering.
\end{restatable}

The two differences are the approximation factor, which is $2w^*$ instead of $w^*$ in the case of $k$-means, and the success probability, which is $0.99$ at each time step individually instead of all time steps simultaneously.
In the following we present the result for $k$-means and  explain at the end how to deal with the $k$-median cost function.

\subsection{Previous Works and Their Limits.}\label{sec:ourTec}

The problem of $k$-median and $k$-means clustering with points in $\R^d$ under (static) differential privacy has been almost closed in the last few years. From the lower-bound perspective, any algorithm must suffer a multiplicative factor $\Omega(1)$, as in the non-private case. 
Table~\ref{tab:ub} summarizes the state-of-the-art results for upper bounds -- stated for $k$-means, with the same picture for $k$-median.
We stress that our algorithm has the same  multiplicative guarantee as the static algorithms of Ghazi et al.~\cite{badih_approximation} and Nguyen~\cite{dpHD}. Minimizing the additive term while maintaining a constant multiplicative approximation is an interesting open challenge, both in the static and the dynamic setting.

\begin{table}[H]
\centering
\caption{Previous work on (static) differentially-private $k$-means. (*) indicates that the algorithm is $(\eps, \delta)$-differentially-private. $w^*$ is the best approximation ratio for non-private $k$-means clustering. 
}
\label{tab:ub}
\begin{tabular}{|l | c | c |}
\hline
  & Multiplicative error & Additive error \\ 
  \hline
  Balcan et al.~\cite{balcan17a} & $O\lpar \log(n)^3 \rpar $ & $(k^2 + d) \cdot \poly \log(n) \cdot \Lambda^2 / \eps$\\
  Stemmer and Kaplan~\cite{StemmerK18} (*) & $O(1)$ & $O(k^{1.01} d^{0.51} + k^{1.5}) \cdot \Lambda^2 / \eps$\\ 
  Ghazi et al.~\cite{badih_approximation}, Nguyen~\cite{dpHD} & $(1+\alpha)w^*$ & $(kd + k^{O(1)}) \log (n)^{O(1)}  \cdot \Lambda^2 / \eps$ \\ 
  Chaturvedi et al.~\cite{anamayclustering} (*) & $O(1)$ & $\tilde{O}(k \sqrt{d}) \cdot \Lambda^2 / \eps$\\
  \hline
\end{tabular}
\end{table}

However, those private algorithms for $k$-means fall short of being accurate under continual observation for two reasons. We highlight them before presenting our algorithm that is designed to work around those challenges.

\phantomsection
(1)\label{chall:1} The first reason is that they all rely on dimension reduction. In the static setting, this is performed as follows: first, project all data points onto a low dimensional subspace with dimension $\hat d$. Then, solve $k$-means privately on this low dimensional subspace, where an exponential dependency in $\hat d$ is affordable, and compute the clusters of the solution. Finally, for each cluster, compute a private mean in the original $\R^d$ space.
By sequential composition, this provides a private algorithm.
Applying this method blindly under continual observation would, however, yield an additive error of $\Omega(T)$ as the number of sequential compositions required for projecting back would be $\Omega(T)$ -- instead of a  single one in the static case. One needs, thus, to provide a new way of projecting back to the original $\R^d$ space tailored to continual observation.

\phantomsection
(2)\label{chall:2} The second hurdle those algorithms face is that, even in low dimensional space, they rely on machinery such as the exponential mechanism that is hard to use under continual observation. For instance, it is not known how to use the exponential mechanism in that setting, which is crucially used in Chaturvedi et al.~\cite{anamayclustering}. All other $\eps$-dp static $k$-means algorithms partition the data set and compute a histogram for the different sets (as in \cite{StemmerK18} and \cite{badih_approximation}).

Under continual observation, it is known how to maintain a histogram maintaining $m$
counters \footnote{A histogram over $\ell$  sets maintains $\ell$ counters that are initially 0 such that at each time step each counter can increase or decrease by at most 1.} with additive error $O(m \cdot \log(T)^2 )$
\cite{ChanSS11, FichtenbergerHO21,cardoso2022differentially,HenzingerSricharanSteiner} and, thus, our idea is to reduce maintaining a differentially private clustering to maintaining such histograms. The main 
challenge is to find a suitable partition for which to maintain the histogram. To avoid any potential privacy-leak 
due to dynamically changing partitions, we want to make the partition independent of the data set.
Therefore, we do not use a partition of the data set, but instead, partition the metric space $\R^d$ into fixed subregions and maintain one counter per subregion in the histogram.

Note that the static $\eps$-DP $k$-median algorithm from Cohen-Addad et al.~\cite{kdd} also proceeds that way: it requires only counts for  sets that are fixed regardless of the input (those are cells of a \textit{quadtree}). However, its multiplicative approximation guarantee is $O(\poly (d) \cdot \log (n))$: it is, thus, heavily reliant on dimension reduction to reduce the dependence of the additive error on $d$,
and has additionally an unfortunate multiplicative dependency in $\log (n)$. 

\textbf{Parallel work.} Epasto, Mukherjee, and Zhong~\cite{concurrent} released very recently a similar work on private $k$-means under dynamic updates. 
Their results are incomparable: their main focus was on designing an algorithm using low memory, which they manage to do in \emph{insertion-only} streams. 
For those streams, they achieve essentially the same approximation guarantees as ours in their Corollary 14. They can also have a smaller additive error, to the price of a large constant multiplicative error, see their Corollary 13.

\subsection*{Our Techniques.}
We focus from now on the "pure" $\eps$-differential privacy setting, and defer the results and discussion of $(\eps, \delta)$-DP to \Cref{sec:weakPrivacy}.

In this work, we present an algorithm designed to overcome those difficulties under continual observation. First, to address \hyperlink{chall:2}{(2)}, we simplify a non-private $O(1)$-approximation algorithm from Mettu and Plaxton \cite{MettuP00} so that it works with a fixed, hierarchical decomposition of $\R^d$ and merely requires counting how many input points are in each cell of the decomposition.
Combining this with any continual, $\eps$-DP histogram algorithm allows to solve $k'$-means under continual observation, for any $k'$, however with additive error exponential in $d$ and  $O(1)$-multiplicative error. The resulting set of centers has been computed in a DP manner and thus, for some large enough value of $k'$, we can treat it as a \emph{coreset}, i.e. a set of points that ``represents'' the input points with regards to $k$-means cost (see \cite{dpHD}). We then run the best-known non-DP $k$-means algorithm on this coreset at every time step.
This results in the following theorem. 

\begin{restatable}{theorem}{optLowDim}\label{thm:optLowDim}
Consider an input stream of additions and deletions of points from $B(0, \Lambda)$ in the metric space $(\R^d, \ell_2)$.
For any $\alpha>0$, there exists an algorithm  that is $\eps$-differentially private under continual observation and computes, for each time step $t$, a solution to $k$-means that has cost at most 
$$(1+\alpha)w^* \cdot \opt_t + \frac{k 2^{O(d)}\log (n)^4 }{\eps}\cdot \log(T)^{3/2} \cdot \Lambda^2,$$
 with probability $0.99$, where $w^*$ is the best approximation ratio for non-private and static $k$-means.
\end{restatable}

To show a similar result without an exponential dependency in $d$, we overcome Challenge \hyperlink{chall:1}{(1)}
\emph{Showing how to make dimension reduction work under continual observation with small additive error is the main contribution of our work.}

This works as follows: we use standard dimension reduction (which conveniently is independent of the input data), run the above algorithm in the resulting low-dimensional metric space, and show how to  lift a low-dimensional clustering back into the original space.
The centers in the low-dimensional space define a clustering, and we use the means of all such clusters in the original $\R^d$ as our solution. 
To compute these means in a DP-manner, we  use the counters for each cell from the $\eps$-DP histogram. 
As the additive error in each mean computation for each high-dimensional center is the sum of the additive error of the counters of all cells of its cluster, we need to make sure that each cluster consists of only $k^{O(1)}$ cells of our decomposition.
We had carefully designed our decomposition so that this is true.
Additionally, to make this work, we need to keep multiple counters per cell, namely the number of points as well as the sum of the points by coordinate. 

Unfortunately, the resulting algorithm is only correct with $O(1)$ probability \emph{per time step} and want to get a constant success probability \emph{over all time steps} .
The standard technique of simply decreasing the failure probability in each time step does not work as the additive error would be polynomial in $T$.
Thus, instead
we boost the success probability of each time step by running $\log (T)$ many low-dimensional clustering algorithms, leading only to an increase of the additive error by $\log(T)^{1/2}$. However, to then pick the best solution out of them is not easy as it requires to evaluate the cost of each solution in a DP-manner. We show how to do this based on additional counters we maintain in our histogram.

\subsection*{Organization of the paper}
In the main body of the paper, we focus on the problem of $k$-means under $\eps$-DP. We first present our basis algorithm, that produces a good approximation given only counts on some regions of the dataset. Then, we show how to use it privately under continual observation, with an additive error scaling with $2^d$, and show \Cref{thm:optLowDim}. We show next how to use dimension reduction to reduce the additive error, and how to evaluate the cost of a solution in order to boost the probability of success and show \Cref{thm:main}. 
We then briefly mention how to extend the result to $k$-median, and show \cref{thm:kmed}. 
Most of the proof are presented in the Appendix, as well as the discussion on $(\eps, \delta)$-differential privacy.

\section{Preliminaries}\label{sec:prelim}
In a metric space $(X, \dist)$, we define the following \emph{cost function}. For any $P, \calS \subset X$, $\cost(P, \calS) := \sum_{p \in P} \dist(p, \calS)^2$.
With a slight abuse of notation, when $P = \{p\}$ or $\calS = \{s\}$ is a singleton, we will write $\cost(p,\calS)$ instead of $\cost(\{p\},\calS)$, and $\cost(P, s)$ instead of $\cost(P, \{s\})$. 

The $k$-means problem in $(X, \dist)$ is the following. We are given a set $P \subseteq X$ of points and a positive integer $k$. The goal is to compute a set of $k$ points $\calS$ that minimizes $\cost(P, \calS)$. We call any set of $k$ points from $X$ a \emph{solution}. The points in a solution are called \emph{centers}. 

An algorithm is a $(\alpha, \beta)$-approximation to the $k$-means problem if its cost is at most $\alpha \opt + \beta$, where $\opt$ is the optimal $k$-means cost on the data set. As mentioned in the introduction, in continual observation, the optimal solution at time $t$ is noted $\opt_t$, $n$ is the maximum size of the data set across time, and $T$ is the total number of time steps.

\paragraph{Assumption.}
As usual in studying clustering under differential privacy, we will assume all points come from the ball $B_d(0, \Lambda) = \lbra x \in \R^d: \|x\|_2 \leq \Lambda\rbra$. Without such an assumption, the lower bound of \cite{anamayclustering} shows that one cannot compute a solution with bounded additive error. For simplicity, we will assume that $\Lambda=1$ in all our proofs. This  is without loss of generality: if the dataset is instead in $B_d(0, \Lambda)$, then rescaling merely yields an additive error multiplied by $\Lambda^2$ for $k$-means. 


\subsection{Net Decomposition}\label{sec:netDec}
A $\delta$-net of $X$ is a subset $Z \subset X$ satisfying the two following properties:
\begin{enumerate}
\item \emph{Packing:} for all distinct $z,z'\in Z$ we have $\dist(z,z')>\delta$
\item \emph{Covering:} for all $x\in X$ we have $\dist(x,Z)\leq \delta$
\end{enumerate}
It is well known that the unit ball in $\R^d$ has a $\delta$-net 
for all $\delta > 0$ \cite{gupta}. For any level $i\in \lbra1,\dots,\lceil \log (n)\rceil\rbra$, let $Z_i$ be a $2^{-(i+1)}$-net of $B(0,1)$. An element of $Z_i$ is a \emph{net point} of \emph{level} $i$. In order to uniquely define the level of a net point, we make the assumption that all $Z_i$ are disjoint. We use the notation $\level(z)$ to denote the level of a net point $z$. For any net point $z\in Z_i$ we define the following notions:
\begin{itemize}
    \item For any $l>0$, the $l$-\textit{neighborhood} of $z$ is $N^l(z):= B(v,l2^{-i})$. The covering property of nets can be restated as follows: for any level $i$ and any point $x$, there exists $z\in Z_i$ such that $x \in N^{0.5}(z)$.
    \item The \emph{children} of $z$ are the net points of level $i+1$ in $N^4(z)$. We stress the fact that the children relationship cannot be represented as a tree: each net point can be the child of many net points from one level higher.
    \item Given an input set $P$, the \textit{value} of $z$ is $\val(z):= 2^{-2i} |N^1(z)\cap P|$.
\end{itemize}

From a classical property of nets, we can deduce the next lemma proven in the appendix~\ref{ap:proofPrelim}

\begin{restatable}{lemma}{sizeNei}\label{lem:sizeNei}
The two following properties hold:
\begin{itemize}
    \item For every $i$, $|Z_i| \leq 2^{id+3}$.
    \item For every point $x$ and $i$, $x$ is in the $1$-neighborhood of at most $2^{2d}$ net points of $Z_i$.
\end{itemize} 

\end{restatable}

\subsection{Differential Privacy}


We start with the formal definition of differential privacy under continual observation.
Let $\sigma$ be a stream of updates to a dataset: $\sigma(t)$ corresponds to either the addition of a new item to the dataset, the removal of one, or a ``no operation" time (where nothing happens to the dataset).  
Two streams are \textit{neighboring} if they differ by the addition of a single item at a single time step (and possibly its subsequent removal). 
This is the so-called \emph{event-level} privacy \cite{dwork2010differential}.
Let $\calM$ be a randomized mechanism that associates a stream $\sigma$ to an output $\calM(\sigma) \in \calX$.
For an $\eps \in \R^+$, $\calM$ is $\eps$-differentially private under continual observation if, for any pair of neighboring streams $\sigma, \sigma'$, and any possible set of outcomes 
$\calS \subseteq \calX$,  it holds that 
$$\Pr\left[ \calM(\sigma) \in \calS\right]  \leq \exp(\eps) \Pr\left[ \calM(\sigma') \in \calS\right].$$

The main tool we will use for the design of private algorithm is a histogram mechanism, that allows to count how many points are in sets. 
Due to space constraint, we define formally this mechanism in \cref{ap:proofPrelim}. Essentially, given a stream of update and ground sets $G_1,..., G_m$, such that each input element is part of at most $b$ sets, one can maintain the count of how many elements are in each $G_i$, privately under continual observation, with additive error $O(b) \cdot \frac{\log(T)}{\eps} \cdot  \lpar \sqrt{\log (T)} + \log(m/\beta)\rpar $. 

We use crucially this histogram construction to prove the following lemma. It allows to count over \textit{clusters} of solutions to $k$-means that evolve over time. 
\begin{restatable}{lemma}{dpCountClusters}
\label{lem:dpCountClusters}
Let $1>\alpha > 0$, and $\sigma$ be a stream of addition or deletion of points from a ground set $G \subseteq \R^d$ of length $T$, and let $X_t$ be the set of points at time $t$ with $1 \le t \le T$. Let $f : G \rightarrow [-1, 1]$ be a function, and for each time $t$ let $\calC^t$ be a set of $k$ centers.
There is an $\eps$-DP algorithm that releases a value $\Sigma_i^t$ for each $i = 1, ..., k$ and time $t$, with the following guarantee:
\begin{itemize}
    \item for any time $t$, there is a partition $C^t_1, ..., C^t_k$ of $X_t$ such that assigning points of $C^i_t$ to the center $\calC^t_i$ yields a $k$-means solution with cost $(1+\alpha) \cost(X_t, \calC^t) + O(1/\alpha)$,
    \item for any time $t$, it holds with probability $1-\beta$ such that
$\left | \Sigma_i^t - \sum_{p \in C^t_i} f(p) \right| \leq  k \cdot \alpha^{-O(d)} \cdot \log(n)^3 \cdot\frac{\log(T)}{\eps} \cdot \lpar \sqrt{\log (T)} + \log(1/\beta)\rpar$.
\end{itemize}
\end{restatable}
In this lemma, the sets of centers $\calC^t$ are arbitrary and can change completely at each time step -- which in particular means we cannot apply standard histogram techniques.  
We also note that the clusters $C^t_1, ..., C^t_k$ do not necessarily corresponds to the optimal clusters for centers $\calC^t$: however, this clustering gives a solution with almost the same cost. We will abuse notations and consider them as the clusters associated with $\calC^t$: the additive term $O(1/\alpha)$ will be negligible compared to other sources of error.
Last, the probability statement is only for the estimation of $\sum_{p \in C^t_i} f(p)$: the construction of the partition $C^t_1, ..., C^t_k$ is deterministic.

\section{A Well-Structured Static Clustering Algorithm}\label{sec:mp}

In this section, we present a non-private algorithm that is \textit{well structured} (in the sense that it can be easily turned into a private algorithm, as we will see), and works in low dimensional space. 
Our algorithm is highly inspired by the static greedy algorithm for $k$-means proposed by Mettu and Plaxton~\cite{MettuP00}. This essentially reduces the problem to computing a histogram, for which there are known algorithms working privately under continual observation.

More precisely, the algorithm computes a solution with a constant multiplicative approximation factor, taking as an input a mere summary of the data: given the recursive net decomposition described in \cref{sec:netDec}, the algorithm merely requires a function $\nval$ that assigns a weight to every net point. This function is an approximation to the value of net points, as formalized in \cref{def:thresh}.

The algorithm works in $k$ rounds: at each round, a new center is placed greedily as follows. Initially, all net points are \emph{available}. 
At the beginning of a round, the algorithm selects the available net point with maximal $\nval$. 
To find the precise location of the next center, the algorithm recursively selects the child of the current net point that has maximum $\nval$, until the last level: the new center is the last net point.
Then, the algorithm removes a number of net points from the set of available ones: namely, all those whose $89$-neighborhood contain the new center.
The algorithm is formally described in \cref{alg:mp}. Note that we didn't attempt to improve the constants appearing, as they only affect marginally the running time.

\begin{restatable}{algorithm}{MetPla}
\caption{$\mettuP(\nval, k)$}
\label{alg:mp}
\begin{algorithmic}[1]
\State{\textbf{Input:} a function $\nval$ that assigns a non-negative number to any net point}
\State{Define the set of \textit{available} net points to be all the net points $\cup_i Z_i$.}
\For{$j$ from $1$ to $k$}
\State{$l \gets 1$}
\State{Let $z_j^1$ be the available net point with largest $\nval(z_j^1)$}
\While{$\level(z) < \lceil \log (n)\rceil $}
\State{Let $z_j^{l+1}$ be the child of $z_j^l$ with largest $\nval \lpar z_j^{l+1}\rpar$}
\State{$l \gets l+1$}
\EndWhile
\State{$c_j \gets z_j^l$}
\For{all available net point $z$}
\If{$c_j \in N^{89}(z)$}
\State{Remove $z$ from the set of available net points}
\EndIf
\EndFor
\EndFor
\State{\textbf{Output:} the solution $\calC = (c_1,\dots,c_k)$}
\end{algorithmic}
\end{restatable}

As explained in the introduction, the algorithm works provided a noisy summary of the input, namely $\nval$. Those will be chosen to approximate the real value of each net point. 

\begin{definition}[$\thresh$-threshold condition]\label{def:thresh}
For any $\thresh \geq 1$, we say that the function $\nval$ verifies the $\thresh$-\textit{threshold condition} if for each net point $z$,
    { \parskip = 0pt
    \begin{itemize}
        \item  when 
$\nval(z) \geq \thresh  $, then $\nval(z) \in [\val(z)/2, 2\cdot  \val(z)]$, and
\item  when
$\nval(z) < \thresh  $, then $\val(z) \leq 2 \cdot \thresh $.
    \end{itemize}
    }
\end{definition}
The utility of this algorithm is stated in the following theorem, proven in \cref{ap:proofMP}
\begin{restatable}{theorem}{mp}\label{thm:mp}
For any threshold $\thresh \geq 1$, if $\nval$ verifies the $\thresh$-threshold condition, then the algorithm $\mettuP(\nval, k)$ computes a $\lpar O(1), O(k\cdot \thresh) \rpar$-approximate solution to the $k$-means problem.
\end{restatable}

\section{Differential Privacy in Low Dimension: Proof of \Cref{thm:optLowDim}}\label{sec:privacyLD}

Using an observation from \cite{dpHD}, it is possible to turn any private constant factor approximation into an algorithm with the best possible non-private approximation ratio, henceforth proving \Cref{thm:optLowDim}. The key idea is that, for any $\alpha > 0$,  any $O(1)$-approximate solution to $k'$-means, with 
$k' = k \alpha^{-O(d)} \log(n/\alpha)$, has cost at most $\alpha \opt_k$, where $\opt_k$ is the optimal cost using $k$ centers. 

 Therefore, given a non-private algorithm $\calA$ with approximation guarantee $w^*$, the following algorithm is private and has multiplicative approximation $(1+\alpha)w^*$, and additive error  $k\cdot \alpha^{-O(d)} \log(n)^4 \cdot\log(T)^{3/2} \cdot{\eps}^{-1}$: (1) We find privately an $O(1)$-approximation $\calS$ to $k'$-clustering (described below), and then (2) we estimate privately the size of each cluster using mechanism $\calM'$ from \Cref{lem:dpCountClusters}. Let $P'$ be the dataset formed by the centers of $\calS$, weighted by the private number of points in their respective cluster. (3) Then, using $\calA$, we compute a $w^*$-approximate solution of $k$-means on $P'$, which is our final output.
\cite{dpHD} formally shows the privacy and utility guarantee of this static algorithm. The privacy stems from the privacy of $\calS$, and the private estimate of the size of each cluster. The computation of $\cal A$ is simply a postprocessing step that does not affect the privacy guarantee. The utility guarantee follows from the triangle inequality: the cost of replacing any point by its center in $\calS$ is negligible compared to the cost of the optimum solution. We make the argument formal in \Cref{sec:optLD}.

\paragraph{$\eps$-DP $O(1)$-approximation algorithm.} Thus, we give an algorithm that computes a constant-factor approximation -- which can be improved using the techniques of \cite{dpHD}. 
For this, we show how to compute a private summary of the data that can be used by \Cref{alg:mp}, namely, we 
want to maintain approximate values of net points that fulfill the $\Theta$-threshold condition, for $\Theta = 3d\cdot 2^{2d} \log (n)^2 \log(T)^{3/2}  \log(1/\beta)/\epsilon$. This allows to compute a constant-factor approximation. 

To do so we merely use standard histograms (see \Cref{lem:dpCountNet} in Appendix). Indeed, the value of a net point can be computed from the number of points in its 1-neighborhood, thus, maintaining privately the number of points in each 1-neighborhood of net points 
with small additive error suffices.

We let $\cal M$ be the histogram algorithm of \Cref{lem:dpCountNet}. At time $t$, the algorithm 
determines in which 1-neighborhood the updated point belongs to
and uses it as a parameter for an update operation in $\cal M$.
Thus, for any net point $z$, $\cal M$ maintains the size of $N^1(z) \cap X_t$, which is
the set of points from the stream present at time $t$ in $N^1(z)$.

\begin{algorithm}
\caption{$\makePrivLowDim(P, \eps)$}
\label{alg:makePriv}
\begin{algorithmic}[1]
\For{each time step $t$ when given update operation $u(p)$ with $u \in \{insert,delete\}$}
    \State{Give $u(p)$ to $\cal M$ which returns for each net point $z$ an estimate $c(z, t)$ of $|N^1(z) \cap X_t|$}
    \State{For each level $\ell$ and net point $z$ at level $\ell$, define $v(z, t)$ to be the value of 
    $z$ at time $t$ computed with counts $c(z,t)$: $v(z, t) :=  2^{-2\ell} c(z, t)$}
    \EndFor
    \State{\textbf{Output:} $v(z, t)$ and $c(z,t)$ for all $z$}
\end{algorithmic}
\end{algorithm}

To show that his algorithm is a constant-factor approximation we show that the values computed with $\makePrivLowDim$ are private and that they satisfy the assumption of \Cref{thm:mp}. Thus, combining this algorithm with \Cref{thm:mp} yields a private constant factor approximation. We formalize this in the following lemma proven in \cref{app:privacyLD}: 
\begin{restatable}{lemma}{linf}\label{thm:linf}
Consider an input stream of additions and deletions of points from $B_d(0, 1)$.
There exists an algorithm  that is $\eps$-differentially private under continual observation and computes, at each time step $t$, a solution to $k$-means that has cost at most 
$$O(1) \cdot \opt_t + k d 2^{2d} \cdot \frac{\log(n)^2\log(T)^{3/2}}{\eps},$$
with probability $0.99$, where  $n$ is the maximum size of the dataset at any given time.
\end{restatable}

\section{Differential Privacy in High Dimensions}\label{sec:privacyHD}

To reduce the additive error to polynomial in $d$,
 we use the following result, that states that the cost of any clustering is roughly preserved when projecting onto $O(\log (k))$ dimensions.
For a set $P$, its average is $\mu(P) :=  {1/|P|}\cdot \sum_{p \in P} p$, and for a projection $\pi : \R^d \rightarrow \R^{\hat d}$  let $\pi(P) := \lbra \pi (p), p \in P \rbra$

\begin{lemma}[Theorem 1.3 in Makarychev, Makarychev, Razenshteyn~\cite{MakarychevMR19}, see also Becchetti et al.~\cite{BecchettiBC0S19}]\label{lem:dimred}
    Fix some $1/4 \geq \alpha > 0$ and $1 > \beta > 0$. There exists a family of random projection $\pi : \R^d \rightarrow \R^{\hat d}$ for some  $\hat d = O(\log (k/\beta) \alpha^{-2})$ such that, for any set $P \in \R^d$, it holds with probability $1-\beta$ that for any partition of $P$ into $k$ parts $P_1, ..., P_k$,
    \[\sum_{j=1}^k \cost\big(\pi(P_j), \mu(\pi(P_j))\big) \in \lpar 1 \pm \alpha\rpar \cdot \frac{\hat d}{d} \cdot \sum_{j=1}^k \cost(P_j, \mu(P_j)).\]
\end{lemma}

This lemma allows us to compute clusters in the low dimensional space, incurring an additive error $\alpha^{-O(\hat d)} = (k/\beta)^{O(\log(1/\alpha))}$   instead of $2^{O(\hat d)}$ using the algorithm from previous section. 
However, an additional step is required to lift the solution from $\R^{\hat d}$ back into the original space $\R^d$  and recover the centers in that space.

Our technique is as follows: we project onto $\R^{\hat d}$ using \Cref{lem:dimred} with success probability $99/100$,  find a clustering in that space using \Cref{thm:optLowDim}, and apply \Cref{lem:dpCountClusters} to that clustering to compute key values, allowing to recover centers in the original $\R^d$ space. Crucially, \Cref{lem:dpCountClusters} is applied only on clusters \emph{from the low-dimensional space}, and the additive error scales only with $2^{O(\hat d)} = k^{O(\log(1/\alpha))}$.
This produces at each time step a correct solution with probability $4/5$. This algorithm fails on many time steps (roughly a fifth of them):  
we show in \Cref{sec:boostProba} how to boost the probability, and give an algorithm that with probability $1-\beta$ is correct at \emph{all} time steps.

\subsection{An Algorithm for $k$-Means in All Dimensions, with Low Success Probability}

\paragraph*{Computing a Private Summary of the Data Points.}
In order to use \Cref{lem:dimred} and project back the computed solution, we slightly change the private summary computed from the points as follows. First, each point is projected onto $\hat d = O(\log(k/\beta))$ dimensions, and a clustering is computed in that space. Using \Cref{lem:dpCountClusters}, one can then for each cluster $C'^t_j$ in $\R^{\hat d}$
estimate the number of points in $C'^t_j$ as well as 
estimate the sum of all points in $C'^t_j$ \emph{using their $d$-dimensional coordinates, i.e., in $\R^d$}. This happens in 
\Cref{alg:hd}. In line 8 of \Cref{alg:hd}, we take the ratio of these two values, which gives the mean of each cluster in the original space $\R^d$ -- which is the optimal $1$-mean solution for this cluster. 
As a technicality, we need to ensure that all projected points are in the ball $B_{\hat d}(0, \log (n))$: this happens with high probability and allows to preserve privacy in the (unlikely) case the projection fails. 
Since the diameter of the projected dataset is now $\log(n)$, all additive errors from the previous section are rescaled by $\log(n)$ (since we assumed the diameter was $1$, for simplicity).

\begin{algorithm}
\caption{$\makePrivate(P, \eps, \alpha, \beta)$}
\label{alg:hd}
\begin{algorithmic}[1]
\State{Let $\pi$ be a dimension reduction given by \Cref{lem:dimred} with success probability $0.99$, and let $\hat d$ be the dimension of $\pi(P)$}
\For{$p \in P$}
\State{define $p'$ to be the projection onto the $\hat d$-dimensional ball $B_{\hat d}(0, \log (n)) $ of the point $\sqrt{\frac{d}{\hat d}} \pi(p)$.}
\EndFor
\State{Let $P'$ be the set of all points $p'$.}
\State{Maintain an approximate solution $\calC^t$ for $k$-means on $P'$, using \Cref{alg:optLowDim} with parameter $\eps / (d+3)$}
\State{Use the algorithm from \Cref{lem:dpCountClusters} with precision parameter $\alpha$, centers $\calC^t$, privacy parameter $\eps/(d+3)$ and probability $\beta / (d+3)$. The outcome of the algorithm is a partition of $P'$ into clusters ${C'}_1^t, ..., {C'}_k^t$, and for each ${C'}_j^t$:
\begin{itemize}
    \item the number of points $n_j^t := |{C'}_j^t|$ , 
    \item the sum $\sumNorm(j, t) := \sum_{p: p' \in {C'}_j^t} \logor p \rnor_2$
    \item for each dimension $i \in \lbra 1,..., d\rbra $, $\sumEst(j, t)_i = \sum_{p: p' \in {C'}_j^t} p_i$
\end{itemize}
}
\State{Define $\sumEst(j,t)$ to be the vector with coordinates $\sumEst(j,t)_i$.}
\State{$\bc_j(t) \gets \frac{\sumEst(j)}{n_j^t}$}
\State{\textbf{Output:} the centers $\bc_j(t)$, the values $n_j^t$, $\sumNorm(j,t)$ and $\sumEst(j,t)$ for all $j \in \lbra 1,..., k\rbra$ and all time $t$.}
\end{algorithmic}
\end{algorithm}

To simplify the equations, we will note in the following
$\calE(\eps, \beta, n) = \frac{\log(n)^2\log(T)}{\eps} \cdot \lpar \log(1/\beta) + \sqrt{\log T}\rpar$.
The value $\sumNorm$ will only be used later in \Cref{sec:boostProba}, but we defined it here already for simplicity.
The proof of the next lemma can be found in \Cref{ap:hd}.
        
\begin{restatable}{lemma}{makePrivHD}
\label{lem:makePrivHD}
    Algorithm~\ref{alg:hd} is $\eps$-differentially private. Furthermore, the two following properties hold:
    \begin{itemize}
        \item for each time $t$, with probability at least $0.99$, the clustering in line 6 is an $\lpar (1+\alpha)w^*, k^{O(1)} \log(n)^2\calE(\eps, 0.01, n)\rpar$-approximation,
        \item for each time $t$ and cluster ${C'}_j^t$ computed line 6, it holds with probability $1-\beta$ that:
            { 
    \parskip=0pt
    \begin{itemize}
        \item $\big|n_j^t - |{C'}_j^t|\big| \leq \calE(\eps, \beta, n)\cdot d \cdot\log (n) \cdot k^{O(1)}$, 
        \item $\logor \sumEst(j,t) - \sum_{p: p' \in {C'}_j^t} p\rnor_2 \leq  O(1) \cdot d^{3/2} \cdot \calE(\eps, \beta, n)\cdot \log (n)\cdot \log(d) \cdot k^{O(1)}$, and
        \item $ \left| \sum_{x \in C_j} \|x\|_2^2 - \sumNorm(j,t)\right| \leq k^{O(1)} d \calE(\eps, \beta, n) \log(n)$
    \end{itemize}
    }
    \end{itemize}
\end{restatable}

\paragraph*{Solving $k$-Means From the Private Summary.}

Using the information of $\sumEst$ and the estimated number of points in each cluster, we can compute an estimate of the location of each cluster's mean in $\R^d$. This is what is done in line 8 of Algorithm~\ref{alg:hd}, and the proof that the additive error remains controlled is in \Cref{lem:hdkmeans}.

\begin{restatable}{lemma}{hdkmeans}
    \label{lem:hdkmeans}
    For any constant $\alpha >0$, the centers
     $\bc_1(t), ..., \bc_k(t)$ computed line 8 of Algorithm~\ref{alg:hd} with failure probability $\beta = 1/10$ form a $((1+\alpha)w^*, ~k^{O(1)} \cdot d^2 \cdot \log(n)^4 \cdot\log(T)^{3/2} \cdot{\eps}^{-1} )$-approximate solution with probability $4/5$. 
\end{restatable}

\subsection{Boosting the probabilities}\label{sec:boostProba}
Algorithm \ref{alg:hd} $\makePrivate$  outputs  centers for each time steps that can be used to
create a clustering which is, with constant probability, a $\lpar (1+\alpha)w^*, ~k^{O(1)} \cdot d^3 \cdot \log(n)^3 \log(T)^{3/2} \cdot{\eps}^{-1} \rpar$-approximate solution. 
This probability is only constant, 
as the dependence of the additive error on the failure probability of the dimension reduction is polynomial:\footnote{This dependency is hidden in \Cref{lem:hdkmeans} as we used a success probability $0.99$. To make it explicit, note that if the target success probability of the dimension reduction is $\beta_d$, then $2^{\hat d} = (k/\beta_d)^{O(1)}$: $1/\beta_d$ appears thus with the same exponent as $k$ in the additive error.} therefore, increasing the success probability naively would dramatically increase the additive error.
However, we want a result that holds for \emph{all} time steps with constant probability. Instead of simply increasing the success probability to $1/T$ and performing a union bound over all steps (which would yield an additive error $\poly(T)$),  we proceed as described in the introduction -- run several copies and output the solution with best cost -- to get the following lemma. \Cref{thm:main} is a direct consequence of this lemma by setting $\kappa = 0.001$.

\begin{restatable}{lemma}{boostProba}\label{lem:boostProba}
    For any $\kappa > 0$, there is an $\eps$-DP algorithm that,  with probability $1-\beta$,  computes for all time steps $t$ simultaneously a $\lpar (1+\alpha)w^*,  k^{O(1)}  \cdot d^2 \log(1/\beta) \log(n)^4 \cdot \log(T)^{3+\kappa} \cdot \eps^{-1} \cdot \kappa \rpar$-approximation to $k$-means. 
\end{restatable}

As mentioned, the main difficulty in proving \Cref{lem:boostProba} is to evaluate the cost of a clustering, privately and for all time steps simultaneously.
Note that it is crucial that the estimated cost is correct at all time $t$ simultaneously: we cannot increase the success probability of this estimates, as we precisely use it to boost the probability of other steps.

To estimate the cost of a single cluster, we rely on the following identity (see \Cref{lem:costVariance}):
for any set $X$, $\cost(X, \mu(X)) =\sum_{x\in X} \|x\|_2^2 - \frac{1}{|X|}\cdot \logor \sum_{y \in X} x\rnor_2^2$.
Therefore, to maintain the clustering cost, it is enough to maintain, for each cluster $C^t_j$,  both (1)  $\sum_{x \in C^t_j} \|x\|_2^2$, which is precisely $\sumNorm(j,t)$, and (2) $\sum_{x \in C^t_j} x$, which is $\sumEst(j, t)$. Those two quantities are maintained by Algorithm \ref{alg:hd}. 
By \Cref{lem:makePrivHD}, $\sumNorm(j,t)$ is an estimate of $\sum_{x \in C^t_j} \|x\|_2^2$, up to an additive error $k^{O(1)} d \calE(\eps, \beta, n) \log(n)$. Therefore, we can show the following lemma:

\begin{restatable}{lemma}{costEst}
    \label{lem:costEst}
    For any time $t$ and any cluster ${C'}^t_j$ from the call to \Cref{lem:dpCountClusters} in \Cref{alg:hd}, define $C^t_j := \lbra p : p' \in {C'}^t_j \rbra$. 
    When the probability parameter of \Cref{alg:hd} is $\beta/T$, then it holds with probability $1-\beta$ that, for all time $t$ and cluster $j \in \{1, ..., k\}$ simultaneously,
    \[\left| \cost(C^t_j, \mu(C^t_j)) - \lpar \sumNorm(j,t) - \frac{\logor \sumEst(j, t)\rnor_2 ^2}{n_j}\rpar\right| \leq k^{O(1)} \cdot \frac{d^{3/2} \cdot \log(n)^3 \log(T) \log(T/\beta)}{\eps}.\]
\end{restatable}

This allows therefore to maintain an approximate estimate of the cost of the solution: thus, it is possible to boost the success probability by running several copies of the algorithm and output one with best estimated cost. We give proof and more details in \Cref{ap:hd}.

\section{Privately Counting on Clusters}\label{sec:dpCountClusters}
This section proves  our key \Cref{lem:dpCountClusters}. For this, we show how to break the ball $B_d(0, 1)$ into \emph{fixed} sets $G_1, ..., G_m$, such that each $C^t_i$ can be described as the union of few $G_j$. In this way, computing the sum $\sum_{p \in C_i^t} f(p)$ is reduced to computing $\sum_{p \in G_j} f(p)$, for which we can apply \Cref{lem:dpCountNet}.

Such a decomposition can be constructed e.g. via net trees \cite{Har-PeledM06}: we describe its construction in \Cref{ap:dpCountClusters}.
Those sets $G_1,..., G_m$ are such that $m = O\lpar n^{d}\rpar$ and any point from $B_d(0, 1)$ appears in $b = \log (n)$ sets. \
Thus, standard counting mechanisms (see \Cref{lem:dpCountNet} in Appendix) apply, with additive error $O\lpar d \cdot \log(n)^2 \cdot \frac{\log(T)}{\eps} \cdot \lpar \sqrt{\log (T)} + \log(1/\beta)\rpar\rpar$. 

Therefore, in order to maintain $\sum_{p \in C_i^t} f(p)$, we only need to show how to express this sum with few $\sum_{p \in G_i \cap X_t} f(p)$. This is done in the next lemma, shown in \Cref{ap:dpCountClusters}.

\begin{restatable}{lemma}{structClustersNet}
    \label{lem:structClusters}
Fix an $1>\alpha > 0$, and $\calG = \lbra G_1,..., G_m\rbra $ as above, and let $\calC = \lbra c_1, ..., c_k\rbra$ be any set of centers. 
Then, it is possible to compute a partition $\calA$ of $B_d(0, 1)$ together with an assignment $a : \calA \rightarrow \calC$ of parts to centers, such that (1) each part of $\calA$ is a set from $\calG$, (2) $|\calA| \leq k\cdot \alpha^{-O(d)} \log(n)$, and (3) for any set $P$,
\[\left|\cost(P, \calC) - \sum_{A \in \calA} \sum_{p \in P \cap A} \dist(p, a(A)) \right| \leq \alpha\cost(P, \calC) + \frac{9}{\alpha}.\]
\end{restatable}

The proof of \Cref{lem:dpCountClusters} follows easily, as sketched before: using standard counting mechanisms (\Cref{lem:dpCountNet}), one can maintain $\sum_{p \in A} f(p)$, for all $A \in \calA$, with additive error $O\lpar d \cdot \log(n)^2 \cdot \frac{\log(T)}{\eps} \cdot \lpar \sqrt{\log (T)} + \log(1/\beta)\rpar\rpar$.
Since each cluster $C_i$ is the union of $k \cdot \alpha^{-2d} \log(n)$ many $A \in \calA$, the additive error is simply multiplied by $k \cdot \alpha^{-O(d)} \log(n)$, which concludes the proof of \Cref{lem:dpCountClusters}.{\setlength{\emergencystretch}{10em}\par}

\section{Extension to $k$-median -- proof of \cref{thm:kmed}}

We now sketch how to prove the following theorem.
\kmed*

For this, we follows the same path as for our $k$-means algorithm. Our algorithm for $k$-means works in three steps: first, project onto a low-dimensional space, and solve in this space. Second, lift the solution up, and third boost the probabilities.
We follow the same path for $k$-median, with a few caveats. 

First, we note that the algorithm for low dimensions from \cref{thm:mp} extends directly to $k$-median: nothing in the argument is special about squared distances -- and the original result from \cite{MettuP00} works for both $k$-means and $k$-median. The techniques from \cite{dpHD} extends directly as well: this yields an algorithm with multiplicative approximation $w^*$ and additive error $2^{O(d)} k \log(n)^2 \frac{\log(T)^{3/2}}{\eps}$.

The dimension-reduction theorem of \cref{lem:dimred} also works for $k$-median. However, given a clustering in the projected space,  computing the centers in the original space is not as easy as for $k$-means. For this, we show that the mean is actually a $2$-approximation of the $1$-median: 

\begin{lemma}\label{lem:meanMed}
    Let $P$ be a set of point in $\R^d$, with optimal $1$-median $m$ and optimal $1$-mean $\mu$. Then, 
    \[\sum_{p\in P} \|p-\mu\| \leq 2 \sum_{p \in P} \|p-m\|.\]
\end{lemma}
\begin{proof}
Our goal is to bound the distance between $m$ and $\mu$. 
    Assume for simplicity that $m$ is the origin, and consider the line $\ell$ going through the origin and $\mu$. 
    Let $P_\ell$ be the set $P$ projected onto $\ell$.
    Since the mean is linear, $\mu$ is also the mean of $P_\ell$, i.e., $\mu = \frac{\sum_{p \in P_\ell} p}{|P|}$.
    Therefore, we have $|P| \|\mu\| \leq \sum_{p \in P_\ell} \|p\|$: since projection only decrease the norm, this is at most $ \sum_{p \in P} \|p\|$. 

    Rephrasing, this means that $\|\mu - m\| \leq \frac{\sum_{p \in P} \|p-m\|}{|P|}$: the distance between the mean and the $1$-median is at most the $1$-median cost divided by $|P|$. Therefore,
    $\sum_{p\in P} \|p-\mu\| \leq \sum_{p\in P} \|p-m\| + |P| \|\mu - m\| \leq 2\sum_{p\in P} \|p-m\|$. This concludes the lemma.
\end{proof}

Therefore, given a clustering in low-dimensional space, using the mean of each cluster in high-dimensional space yields a mere factor $2$ additional loss. Techniques developed in \cref{sec:privacyHD} allow to compute privately the mean under continual observation.

The third step of our $k$-means algorithm, boosting the probabilities, requires to be able to compute privately under continual observation the cost of a $k$-means solution: for this we used \cref{lem:costEst}. 
How to extend to $k$-median, and therefore boosting the probabilities, is interesting future work. Without the third step we get the result stated in \cref{thm:kmed}.

We additionally note that the simple \cref{lem:meanMed} can be applied to other settings, e.g., to extend the result from \cite{anamayclustering} to get a locally-private algorithm for $k$-median as well.

\section{Conclusion}
We present the first $\eps$-differentially private $k$-means clustering algorithm under continual observation. It has almost the same multiplicative error as the best static non-private algorithm and an additive error of $\poly(k, d, \log (n), \log (T)) \Lambda^2/\epsilon$, i.e., unlike a naive approach it is \emph{not} exponential in $d$ and \emph{not} linear in $T$. 

Our techniques extend to $k$-median, with a slightly worsen multiplicative approximation guarantee, as well as with a weaker success probability. Filling the gap is an interesting open question. 
\bibliography{biblio}

\newpage
\appendix
\section{Extended Preliminaries \Cref{sec:prelim}}\label{ap:proofPrelim}

In this section, we will prove Lemma~\ref{lem:sizeNei} and introduce two additional results essential for subsequent proofs.

\sizeNei*

\begin{proof}
    The Lemma directly follows from this classical property of nets that we prove for completeness:\\

    For any $\delta > 0$, any $\delta$-net $Z$ of a set $X \subset \mathbb{R}^d$, any point $x \in X$, and any $l > 0$, the ball $B(x, l \cdot \delta)$ contains at most $2^d (l + 0.5)^d$ points of $Z$.\\

    By the packing property of nets, any ball of radius $\delta/2$ around a point in $Z \cap B(x, l \cdot \delta)$ is disjoint from other such balls. Furthermore, these balls are entirely contained within $B(x, (l+0.5) \cdot \delta)$. The volume of a ball of radius $\delta/2$ is a $(1/(2(l+0.5))^d)$ fraction of the volume of a ball of radius $(l+0.5) \cdot \delta$. Therefore, we can derive the inequality $|Z \cap B(x, l \cdot \delta)| \leq 2^d (l+0.5)^d$.
\end{proof}

To deal with squared distances, we will also use generalization of triangle inequality:
\begin{fact}[Squared Triangle Inequality]
\label{lem:weaktri}
Let $a,b,c$ be an arbitrary set of points in a metric space with distance function $d$. Then for any $\alpha>0$
$$d(a,b)^2 \leq (1+\alpha) d(a,c)^2 + \left(\frac{1+\alpha}{\alpha}\right) d(b,c)^2 $$ 
\end{fact}
\begin{proof}
    By the regular triangle inequality, $d(a,b)^2 \leq (d(a, c) + d(c, b))^2$. We note $x=d(a, c), y = d(b, c)$, and observe that $\alpha x^2 + y^2\alpha  - 2 x y = (\sqrt \alpha x - y/\sqrt{\alpha} )^2 \geq 0$, and therefore 
    $(x+y)^2 \leq (x+y)^2 + \alpha x^2 + y^2/\alpha  - 2 x y = (1+\alpha)x^2 + (1+1/\alpha) y^2$, which concludes.
\end{proof}

The main tool we will use for the design of private algorithm is the following mechanism for histograms, it is a slight extension of Lemma 11 in~\cite{FichtenbergerHO21}:

\begin{restatable}{lemma}{dpCountNet}
\label{lem:dpCountNet}
Let $\sigma$ be a stream  of length $T$ of addition or deletion of points from a ground set $G$, and let $P_t$ be the set of points at time $t$ with $1 \le t \le T$. Let $f : G \rightarrow [-1, 1]$ be a function, and let $G_1, ..., G_m \subseteq G$ be such that at any $p \in G$ is included in at most $b$ $G_i$.
There is an $\eps$-DP algorithm that, for all time step individually, with  probability $1-\beta$, estimates for each set $G_i$  the value of $\sum_{p \in G_i \cap P_t} f(p)$, with additive error 
$O(b) \cdot \frac{\log(T)}{\eps} \cdot  \lpar \sqrt{\log (T)} + \log(m/\beta)\rpar $. 
\end{restatable}

\begin{proof}
Let $\eps' = \eps/b$.
Given a sequence of numbers from $[-L, L]$ a \emph{continual counting mechanism} outputs, after the arrival of each number $x_i$, an estimate of the sum $\sum_{i=1}^t x_i$.  
The binary mechanism is 
an $\eps'$-differentially private continual counting mechanism that estimates the sum $\sum_{i=1}^t x_i$.
With probability at least $1-\beta$ the additive error
of all time steps 
is  $O(L\eps'^{-1} \log(T)^{2} \log(1/\beta))$
~\cite{dwork2010differential,ChanSS11,FichtenbergerHO21}.
Note that the proof given in these references only shows the
claim for numbers from $\{-L, -L+1, \dots, L-1, L\}$, but the same proof goes through verbatim for numbers from $[-L,  L]$. Further note that the additive error bound in these references is stated only on a per-time step basis, i.e., with probability at least $1-\delta$, the additive error at any \emph{individual} time step is
$O(L\eps'^{-1} \log(T)^{3/2} \log (1/\delta)).$
However, as observed in~\cite{jain2021price} the bound on the additive error for any individual time step can be improved to
$O(\eps'^{-1} \log(T) ( \sqrt{\log (T)} + \log (1/\delta)))$
using the first part of Corollary 2.9 of \cite{ChanSS11}.

To prove the lemma we run the $\eps'$-differentially private binary mechanism separately for each set $G_i$ for $i = 1, \dots, m$.  
This is identical to running a histogram algorithm where each set $G_i$ corresponds to a counter in the histogram algorithm. We give the full analysis here, as only the case $b=1$ has been studied
in the literature on $\eps$-differentially algorithms (but $b \geq 1$  has been studied for $(\eps,\delta)$-differentially private algorithm,s using a different algorithm), and we crucially rely on the fact that $b$ is small in our analysis.

Whenever a point $p \in G$ is inserted, $f(p)$ is added to the binary mechanism for all the sets $p$ belongs to, whenever it is deleted $-f(p)$ is added to the binary mechanism for all the sets $p$ belongs to. Between any two updates, the binary mechanism outputs the same sum estimate no matter how often it is queried. Thus, we can analyze the privacy loss as if the mechanism outputted exactly one estimate after each update.

\emph{Privacy analysis.} The data from each update operation is only used for $b$ $\eps'$-differentially mechanism. By the composition theorem (Theorem 3.16 in \cite{dwork2014algorithmic})
it follows that it 
is $b \cdot \eps' = \eps$-differential private.

\emph{Utility analysis.}
We set $\delta = \beta/m$. It follows that, for any time step $t$, with probability at least $1-\beta$ for all $m$ binary mechanisms the additive error is $O(L \cdot \eps'^{-1} \log(T) ( \sqrt{\log (T)} + \log (1/\delta))) = $
$O(L\cdot  b \cdot \eps^{-1} \log (T) (\sqrt{\log (T)} + \log (m/\beta)))$
\end{proof}

Note that, in the previous mechanism, the sets $G_1, ..., G_m$ are \emph{fixed} and cannot depend on the time $t$. 
This means that we cannot use it directly to estimate, say, the size of a cluster at time $t$ -- even if we can compute privately the center of each cluster. 
As estimating such quantities is crucial for our algorithms, we will present a generic mechanism to do so. It is stated in Lemma~\ref{lem:dpCountClusters}, we defer the proof to \Cref{sec:dpCountClusters}. The basic idea is to compute at each time $t$ the size of each cluster as the sum of $k \cdot 2^{O(d)} \log (n)$ many sets of the fixed partition $G_1, ..., G_m$.

\section{Missing Proofs for the Well-Structured Static Clustering Algorithm (Section~\ref{sec:mp})}\label{ap:proofMP}
This section is devoted to the proof of Theorem~\ref{thm:mp}. For the sake of clarity, we will restate the algorithm and the theorem here.

\MetPla*

\mp*

Our first remark on the algorithm is the following: at the time where a net point $z_j^l$ is selected by the algorithm in line 7, it is available.
\begin{restatable}{lemma}{available}
\label{lem:available}
At any moment of the algorithm, if a net point is available, its children are also available.
\end{restatable}
\begin{proof}
Let $z_1$ be a cell of level $i$, and let $z_2$ be a child of $z_1$. We have by definition $z_2\in N^4(z_1)$ and $\dist(z_1,z_2)\leq 4\cdot 2^{-i}$. We will prove that if $z_2$ becomes unavailable, then $z_1$ also becomes unavailable. 

Suppose that a new center $c_j$ is selected, and $z_2$ is no longer available. This means that $c_j \in N^{89}(z_2)$,and $\dist(z_2,c_j) \leq 89 \cdot 2^{-(i+1)} =44.5 \cdot 2^{-i}$. Now using the triangle inequality we have 
$$\dist(z_1,c_j) \leq \dist(z_1,z_2) + \dist(z_2,c_j) \leq 4\cdot 2^{-i} + 44.5 \cdot 2^{-i} = 48.5 \cdot 2^{-i}$$
Therefore $c_j \in N^{48.5}(z_1) \subset N^{89}(z_1)$ and $z_1$ also becomes unavailable when $c_j$ is selected.
\end{proof}

In what follows, we fix a solution $\Gamma$  with $k$ centers. Our goal is to compare the  cost of the solution $\calC$ output by the algorithm and the cost of $\Gamma$. For $\gamma \in \Gamma$, let $P_\gamma$ be $\gamma$'s cluster, namely all points of $P$ assigned to $\gamma$ in the solution $\Gamma$. For this, we analyze the cost of each cluster $P_\gamma$ independently, as follows. 

We split $\Gamma$ into two parts: $\Gamma_1$ is the set of $\gamma\in \Gamma$ such that there exists a net point $z$ of level $\lceil \log(n) \rceil$ such that $\gamma \in N^{1/2}(z)$ and $z$ is still available at the end of the algorithm, and $\Gamma_2 = \Gamma-\Gamma_1$. Centers in $\Gamma_2$ are very close to centers in $\calC$ (at distance at most $O(1/n)$), and the cost of their cluster is therefore well captured by $\calC$.  The bulk of the work is to show that clusters in $\Gamma_1$ are also well approximated by $\calC$.

For $\gamma \in \Gamma_1$ we define $z_\gamma$ to be a net point of the smallest level such that $\gamma \in N^{1/2}(z_\gamma)$ and $z_\gamma$ is still available at the end of the algorithm. We split $P_\gamma$ into two parts: $In(P_\gamma):= P_\gamma \cap N^1(z_\gamma)$ and $Out(P_\gamma) = P_\gamma - In(P_\gamma)$. 

By definition of $z_\gamma$, we know there exists a center of $\calC$ "not too far" from $z_\gamma$. This allows us to bound the cost of $Out(P_\gamma)$ in the solution $\calC$. Furthermore, we can relate the cost of $In(P_\gamma)$ to the value of $z_\gamma$, as done in the following lemma.

\begin{restatable}{lemma}{inAndOut}\label{lem:inAndout}
For all $\gamma \in \Gamma_1$, we have:
    \begin{align}
        &\cost(In(P_\gamma),\calC)  \leq  181^2 \cdot \val(z_\gamma)\label{lem:In}\\
        &\cost(Out(P_\gamma), \calC) \leq 359^2 \cdot  \cost(Out(P_\gamma), \Gamma) \label{lem:Out}
    \end{align}
And for all $\gamma \in \Gamma_2$, we have:
\begin{equation}
    \cost(P_\gamma, \calC) \leq 2\cost(P_\gamma, \Gamma)+2|P_\gamma| (90/n)^2. \label{lem:notDefined} 
\end{equation}
\end{restatable}

\begin{proof}
We start by considering the case where $\gamma$ belongs to $\Gamma_1$. Let $z^*_\gamma$ be net point of level $\level(z_\gamma)-1$ such that $\gamma \in N^{0.5}(z^*_\gamma)$. By definition of $z_\gamma$, the net point $z^*_\gamma$ is not available. Therefore, there is a point of $\calC$ in $N^{89}(z^*_\gamma)$. We have $\dist(z^*_\gamma, \calC) \leq 89 \cdot 2^{-\level(z^*_\gamma)} = 178\cdot 2^{-\level(z_\gamma)}$, and $\dist(\gamma, z^*_\gamma) \leq 0.5 \cdot 2^{-\level(z^*_\gamma)} = 2^{-\level(z_\gamma)}$ because $\gamma \in N^{0.5}(z^*_\gamma)$. This implies by the triangle inequality that $\dist(\gamma, \calC) \leq 179 \cdot 2^{-\level(z_\gamma)}$. 

Now for any $x\in P_\gamma$, we have 
\begin{equation*}
    \cost(x, \calC) = \dist(x,\calC)^2 \leq  (\dist(\gamma,\calC) + \dist(\gamma,x))^2
\end{equation*}

We can now bound the cost of $In(P_\gamma)$ and prove \Cref{lem:In}. If $x \in N^1(z_\gamma)$, we have  $\dist(\gamma,x) \leq \dist(\gamma, z_\gamma) + \dist(z_\gamma, x) \leq 0.5\cdot 2^{-\level(z_\gamma)} + 2^{-\level(z_\gamma)} \leq 2 \cdot 2^{-\level(z_\gamma)} $, and therefore $\cost(x, \calC) \leq (179 \cdot 2^{-\level(z_\gamma)} + 2\cdot 2^{-\level(z_\gamma)})^2 = 181^2 \cdot 2^{-2\level(z_\gamma)} $. Summing this inequality over all $x \in In(P_\gamma)$, we get $$\cost(In(P_\gamma), \calC) \leq \sum_{x \in N^1(z_\gamma)\cap P} 181^2  \cdot  2^{-2\level(z_\gamma)} = 181^2  \val(z_\gamma).$$

We turn to $Out(P_\gamma)$ and show \Cref{lem:Out}. We have $\dist(\gamma,z_\gamma) \leq 0.5 \cdot 2^{-\level(z_\gamma)}$, so if $x$ is outside $N^1(z_\gamma)$, we have $\dist(\gamma,x) \geq 0.5 \cdot 2^{-\level(z_\gamma)}$. Hence
    \begin{align*}
        \cost(x, \calC) &\leq  (\dist(\gamma,\calC)  + \dist(\gamma,x))^2  \\
        &\leq (179 \cdot 2^{-\level(z_\gamma)} + \dist(\gamma,x))^2 \\
        &\leq (2\cdot 179 \dist(\gamma,x) + \dist(\gamma,x))^2 \\
        &=  359 ^2 \dist(\gamma,x)^2 = 359 ^2 \cost(x, \Gamma)
    \end{align*}
    Summing this inequality over all $x \in Out(P_\gamma)$, we get 
    $$\cost(Out(P_\gamma), \calC) \leq 359 ^2 \cost(Out(P_\gamma), \Gamma).$$

Otherwise if $\gamma \in \Gamma_2$, let $z$ be a net point of level $\lceil \log(n) \rceil $ such that $\gamma \in N^{0.5}(z)$. By definition of $\Gamma_2$, $z$ is unavailable at the end of the algorithm, therefore there exists a center $c_j\in \calC$ such that $c_j \subset N^{89}(z)$. This allows us to bound the distance between $c_j$ and $\gamma$. We have $\dist(c_j, \gamma) \leq \dist(c_j,z) + \dist(z,\gamma) \leq 89 \cdot 2^{-\lceil \log(n) \rceil} + 0.5 \cdot 2^{-\lceil \log(n) \rceil} \leq \frac{90}{n}$, so for all $x\in X$:
\begin{align*}
    \cost(x, \calC) &\leq (\dist(\gamma,x) + \dist(c_j,\gamma))^2\\
    &\leq 2 \dist(\gamma,x)^2+2(90/n)^2
\end{align*}
Summing this inequality over all $x\in P_\gamma$, we get
\begin{align*}
    \cost(P_\gamma, \calC) &\leq 2\cost(P_\gamma, \Gamma)+2|P_\gamma| (90/n)^2. \qedhere
\end{align*}
\end{proof}

\subsection{An Easy Instructive Case.}\label{sec:easy}

We now show how to conclude the proof of \Cref{thm:mp} based on \Cref{lem:inAndout}, under the following assumption.
Suppose all net points that are available at the end of the algorithm have a value  less than $2\cdot \thresh $. Then we can directly use \Cref{lem:inAndout} to conclude the proof of  Theorem \ref{thm:mp}. Although unrealistic, this provides a good intuition on our proof of \Cref{thm:mp}.

Indeed, summing \eqref{lem:Out} for all $\gamma \in \Gamma_1$, we get (and this holds regardless of the assumption):
\begin{align*}
    \cost(\bigcup_{\gamma \in \Gamma_1}Out(P_\gamma), \calC) &= \sum_{\gamma \in \Gamma_1} \cost(Out(P_\gamma), \calC) \leq 359^2 \cdot  \sum_{\gamma \in \Gamma_1}\cost(Out(P_\gamma), \Gamma)\\
      &\leq  359^2 \cdot \cost(P, \Gamma)
\end{align*}

Summing \ref{lem:notDefined} for all $\gamma \in \Gamma_2$, we get :
\begin{align*}
    \cost(\bigcup_{\gamma \in \Gamma_2}P_\gamma, \calC) &\leq 2\cdot \cost(P, \Gamma) +  n \cdot (90/n)^2\\
    &\leq 2\cdot \cost(P, \Gamma) + 90^2
\end{align*}

All net points $z_\gamma$ are still available at the end of the algorithm, and with our assumption have a value at most $2 \cdot \thresh $. Summing \eqref{lem:In} for all $\gamma \in \Gamma_1$, we get:
\begin{align*}
    \cost(\bigcup_{\gamma \in \Gamma_1}In(P_\gamma), \calC) &= \sum_{\gamma \in \Gamma_1}\cost(In(P_\gamma), \calC) \leq 181^2 \cdot \sum_{\gamma \in \Gamma_1}\val(z_\gamma)\\
    &\leq 181^2 \cdot \sum_{\gamma \in \Gamma_1} 2 \thresh 
    \leq 181^2 \cdot 2k\thresh 
\end{align*}

Summing the three part, we get $\cost(P, \calC) \leq (359^2+2)\cost(P, \Gamma) + 181^2 \cdot 2k\thresh +90^2$ and Theorem \ref{thm:mp} is proven (recall that $\thresh \geq 1$).

In the rest of the proof, we show how to bound $\sum_{\gamma \in \Gamma_1}\val(z_\gamma)$ without the need for any assumption.

\subsection{Bounding the Values.}

To do so, we first relate the value of a net point to the cost of its 1-neighborhood, in the solution $\Gamma$. We say that a net point $z$ is \textit{covered} by $\Gamma$ if $\Gamma \cap N^2(z) \neq \emptyset$.

\begin{lemma}\label{lem:cover}
    If a net point $z$ is not covered by $\Gamma$, then $ \cost (N^1(z)\cap P, \Gamma) \geq \val(z)$.
\end{lemma}
\begin{proof}
    Let $x\in N^1(z)\cap P$. Since there is no point of $\Gamma$ in $N^2(z)$,  $\dist(x,\Gamma) \geq 2^{-\level(z)}$ and thus $\cost(x,\Gamma) \ge 2^{-2\level(z)}$. Summing this inequality over all $x \in N^1(z)\cap P$, we get $$\cost(N^1(z)\cap P,\Gamma) = \sum_{x\in N^1(z)} \cost(x,\Gamma) \geq |P\cap N^1(z)| \cdot 2^{-2\level(z)}  = \val(z)\qedhere$$
\end{proof}

We can now provide a bound on the sum of values, based on the previous lemma. For this, we aim at matching each $z_\gamma$ (for $\gamma \in \Gamma_1$) with a net point that is not covered, such that the 1-neighborhood of those points is disjoint. Then, it will be possible to bound the sum of values using the upper bound of \Cref{lem:cover}. In order to define the matching, we need the following lemma.

\begin{restatable}{lemma}{pruning}\label{lem:pruning}    Assume that, at the end of the algorithm, there is a net point available with a value greater than $2\cdot \thresh$. Then, for all $j\in \{1,\dots,k\}$ and each of the sequences $(z_j^l)_l$ defined in \Cref{alg:mp}, there exists an index $l_j$ such that:
    {\parskip = 0pt
\begin{itemize}
    \item $\nval(z_j^{l_j})$ is greater than the maximum noisy value of a point available at the end of the algorithm,
    \item for all $j,j'$, for all $l\geq l_j$, $l' \geq l_{j'}$  $N^2(z_j^{l}) \cap N^2(z_{j'}^{l'}) = \emptyset$
\end{itemize}
}
\end{restatable}
\begin{proof}[Proof sketch, see \Cref{ap:pruning}]
    The proof is inspired by the original one from Mettu and Plaxton~\cite{MettuP00}. We sketch it here very briefly and refer to \Cref{ap:pruning} for more details. 
    We describe a procedure to compute the $l_j$, starting with $l_j = 1$ for all $j$, and note that with this choice the first condition is verified by the design of the algorithm: when $z_j^1$ is picked it (on line 5 of Algorithm~\ref{alg:mp}), it maximizes the noisy value among the available net points. 
To enforce the second condition, we proceed as follows: as long as we can find $j<j'$ and $l\geq l_j$, $l' \geq l_{j'}$ such that $N^2(z_j^{l}) \cap N^2(z_{j'}^{l'}) \neq \emptyset$, we update $l_j = l_j + 1$. 

We show in \Cref{ap:pruning} that this procedure is well defined, that it preserves by induction the first property and therefore when it terminates, the two conditions are satisfied.
\end{proof}

We are now ready to prove \Cref{thm:mp}. Our goal is to find $k$ net points not covered by $\Gamma$ with disjoint $1$-neighborhood and with values larger than those of the $z_\gamma$. Then, applying \Cref{lem:cover} would allow to conclude $\cost(\bigcup_{\gamma\in\Gamma_1} In(P_\gamma), \calC) \leq O(1) \cdot \sum_\gamma \val(z_\gamma) \leq O(1) \cdot \cost(P, \Gamma)$. 
We explain in the next proof how to compute those net points.

\begin{proof}[Proof of \Cref{thm:mp}]
In order to bound the sum of values of $z_\gamma$, we will apply the previous \Cref{lem:cover} to $k$ high-value net points not covered by $\Gamma$. Specifically, we will define a function $\phi$ that maps centers of $\Gamma_1$ to net points such that, for all $\gamma \in \Gamma_1$:
{\parskip = 0pt
\begin{itemize}
    \item  for all $\gamma' \in \Gamma_1$ with $\gamma \neq \gamma'$, $N^1(\phi(\gamma))\cap N^1(\phi(\gamma')) = \emptyset$,
    \item  $\phi(\gamma)$ is not covered by $\Gamma$,
    \item  the value of $z_\gamma$ is less than $\max(2\cdot \thresh , \val(\phi(\gamma)) / 4)$.
\end{itemize}
}

Given such a matching $\phi$, we can conclude as follows. 
Using \Cref{lem:inAndout} and the proof of \Cref{sec:easy} we get:
\begin{align*}
\cost(P, \calC) &= \cost(\bigcup_{\gamma \in \Gamma_1}Out(P_\gamma), \calC) + \cost(\bigcup_{\gamma \in \Gamma_1}In(P_\gamma), \calC) + \cost(\bigcup_{\gamma \in \Gamma_2}P_\gamma, \calC)\\
&\leq 359^2 \cost(, \Gamma) +90^2 + 181^2 \sum_{\gamma \in \Gamma_1} \val(z_\gamma)
\end{align*}
Summing the inequality of the third property of $\phi$ gives $\sum_{\gamma \in \Gamma_1}\val(z_\gamma) \leq 2k\cdot\thresh + \sum_{\gamma \in \Gamma_1} \val(\phi(\gamma)))$.
Combined with \Cref{lem:cover} and the fact that the $N^1(\phi(\gamma))$ are disjoints, we can therefore conclude:
\begin{align*}
    \cost(P, \calC)  &\leq 359^2 \cost(P, \Gamma) +90^2 + 181^2 \cdot \lpar 2k\cdot\thresh + \sum_{\gamma \in \Gamma_1}\cost(N^1(\phi(\gamma))\cap P, \Gamma)\rpar \\
    &\leq 359^2 \cost(P, \Gamma) +90^2 + 181^2 \cdot \lpar 2k\cdot\thresh +\cost(\bigcup_{\gamma\in\Gamma_1}N^1(\phi(\gamma))\cap P, \Gamma)\rpar\\
    & \leq 359^2 \cost(P, \Gamma) +90^2 + 181^2 \cdot \lpar 2k\cdot\thresh + \cost(P,\Gamma)\rpar\\
    &\leq  O(1)(\cost(P, \Gamma) + 2k\cdot\thresh).
\end{align*}

The rest of the proof is dedicated to the construction of a matching $\phi$ with the three desired properties. We construct a more general function $\Tilde{\phi}$ that maps centers of $\Gamma$ to net points. The restriction of $\Tilde{\phi}$ to $\Gamma_1$ will verify the desired properties. Let $l_j$ be the indices provided by \Cref{lem:pruning}. We have three cases: for all $j$ such that $c_j$ is covered, let $\gamma^2_j\in \Gamma$ be an arbitrary element covering it. For any net point $z$ of level $\lceil \log(n) \rceil$ such that $\gamma^2_j\in N^{0.5}(z)$, we have by the triangle inequality $\dist(c_j,z) \leq \dist(c_j,\gamma^2_j) + \dist(\gamma^2_j,z) \leq 2\cdot  2^{-\lceil \log(n) \rceil} + 0.5 \cdot 2^{-\lceil \log(n) \rceil} = 2.5\cdot  2^{-\lceil \log(n) \rceil}$. In particular, $z$ is unavailable at the end of the algorithm and $\gamma^2_j\in \Gamma_2$ by definition of $\Gamma_2$. We define $\Tilde{\phi}(\gamma^2_j) = z_j^{l_j}$. 

Otherwise for all $j$ such that at least one of the net points of the sequence $(z_j^{l})_{l\geq l_j}$ is covered but not the last one. We define $\lambda_j \geq l_j$ to be the smallest index such that for all $l\geq \lambda_j$ $z_j^{l}$ is not covered, and $\gamma^1_j$ to be an arbitrary element of $\Gamma$ that covers $z_j^{\lambda_j-1}$. 
We define $\Tilde{\phi}(\gamma^1_j) = z_j^{\lambda_j}$.
$\Tilde{\phi}$ is completed to be an arbitrary one-to-one matching between all remaining $\gamma$ and all $z_j^{l_j}$, for $j$ such that none of the net points of the sequence $(z_j^{l})_{l\geq l_j}$ are covered.

Note that the second item of lemma \ref{lem:pruning} guarantees that if $\gamma$ covers a net point of a sequence, it can not cover a net point of another sequence: this ensures  that our definition of $\Tilde{\phi}$ is consistent. We can now verify that $\phi$ defined as the restriction of $\Tilde{\phi}$ to $\Gamma_1$ verifies the three desired properties.
\begin{itemize}

    \item For all $j$, since $\lambda_j \geq l_j$, \Cref{lem:pruning} ensures that for all $\gamma \neq \gamma'$ $N^2(\Tilde{\phi}(\gamma))\cap N^2(\Tilde{\phi}(\gamma')) = \emptyset$. In particular $N^1(\Tilde{\phi}(\gamma))\cap N^1(\Tilde{\phi}(\gamma')) = \emptyset$.
    \item If $\gamma \in \Gamma_1$, $\Tilde{\phi}(\gamma)$ is not covered by $\Gamma$ by construction.
    \item For the third item, we distinguish three cases. For all $\gamma \in \Gamma$:
\begin{itemize}
    \item Either $\gamma$ is one of the $\gamma^1_j$. In that case, let $z$ be a net point of same level as  $\Tilde{\phi}(\gamma)$ such that $\dist(z,\gamma^1_j) \leq 2^{-(\level(z)+1)}$. $\gamma^1_j$ is covering $z_j^{\lambda_j-1}$ so $z$ could have been chosen by the algorithm instead of $\Tilde{\phi}(\gamma)$, therefore $\nval(C) \leq \nval(\Tilde{\phi}(\gamma))$. 
    Furthermore, $N^1(z_\gamma) \subset N^1(z)$. Therefore, $\val(z_\gamma) \leq \val(z)$. By the $\thresh$-threshold condition, if $\nval(z) < \thresh $ then $$\val(z_\gamma) \leq \val(z) \leq 2\cdot \thresh .$$ Otherwise if $\nval(z)  \geq \thresh$ then $$\val(z_\gamma) \leq \val(z) \leq \nval(z)/2 \le \nval(\Tilde{\phi}(\gamma)) /2  \leq  \val(\Tilde{\phi}(\gamma))/4.$$

    \item Or $\gamma$ is one of the $\gamma^2_j$. In that case, $\Tilde{\phi}(\gamma)$ does not need to verify the third property because $\gamma_j^2 \in \Gamma_2$
    \item Or $\gamma$ is not one of the $\gamma^1_j$ or $\gamma^2_j$, then $\Tilde{\phi}(\gamma)$ is of the form $z_j^{l_j}$. By the lemma \ref{lem:pruning}, $\nval(z_\gamma) \leq \nval(z_j^{l_j})$ because $z_\gamma$ is available at the end of the algorithm. Using the $\thresh$-threshold condition, we have 
    \begin{align*}
    \val(z_\gamma) &\leq \nval(z_\gamma) / 2 \leq \nval(z_j^{l_j}) / 2 \leq \val(z_j^{l_j})/4.\qedhere
    \end{align*}
\end{itemize}
\end{itemize}

\end{proof}

\subsection{Properties of Algorithm~\ref{alg:mp}: Proof of Lemma~\ref{lem:pruning}}\label{ap:pruning}
Before proving \Cref{lem:pruning}, we show some preliminary results.
We begin by establishing a simple Lemma that sets a limit on the distance between two net points that are selected during the same execution of loop line 6 of \Cref{alg:mp}.
\begin{restatable}{lemma}{8diam}
    \label{lem:8diam}
    Let $z_j^l$ and $z_j^{l'}$ be two net points selected by the algorithm at the $j$-th iteration of the loop line 6 of \Cref{alg:mp}, with $l' > l$. Then $\dist(z_j^l,z_j^{l'}) \leq 8\cdot 2^{-\level(z_j^{l})}$.
\end{restatable}
\begin{proof}
The distance between $z_j^{l}$ and one of its child is at most $4 \cdot 2^{-\level(z_j^{l})}$. By induction, we have:
\begin{align*}
\dist(z_j^l,z_j^{l'}) &\leq \sum_{i=0}^{l'-l-1}\frac{4 \cdot 2^{-\level(z_j^{l})}}{2^i} \leq 8\cdot 2^{-\level(z_j^{l})}\qedhere
\end{align*}
\end{proof}

The next Lemma is the key to prove that the procedure to compute the $l_j$'s terminates and verifies the conditions of \Cref{lem:pruning}
\begin{restatable}{lemma}{pruningOneStep}\label{lem:pruningOneStep}
If there exists at the end of the algorithm an available net point with value greater than $2\cdot \thresh$, then for every $j,j',l,l'$ such that $N^2(z_j^l)\cap N^2(z_{j'}^{l'}) \neq \emptyset$ and $j<j'$,
{ \parskip = 0pt
\begin{itemize}
\item $\level(z_{j'}^{1}) \geq \level(z_{j}^{l})+4$
\item $\nval(z_j^{l+1}) \geq \nval(z_{j'}^1)$
\end{itemize}
}
\end{restatable}
\begin{proof}
Let $j,j',l,l'$ such that $N^2(z_j^l)\cap N^2(z_{j'}^{l'}) \neq \emptyset$ and $j<j'$. We start by proving the first point by contradiction: suppose that $\level(z_{j'}^1) \leq \level(z_{j}^{l})+3$. We extend the sequence starting from $z_{j'}^{l'}$ to prove the existence of a "descendant" of $z_{j'}^1$ of level  $\level(z_{j}^{l})+3$ in $N^3(z_j^l)$. We will then prove that this descendant became unavailable when $c_j$ was selected, contradicting Lemma \ref{lem:available}.

More precisely, we pick recursively a sequence $y_0,\dots, y_{i^*}$ such that $y_0 = z_{j'}^{l'}$ and $y_{i+1}$ is an arbitrary child of $y_i$ such that in $N^2(z_j^l) \cap N^2(y_{i+1}) \neq \emptyset $ and $\level(y_{i^*}) \geq \level(z_{j}^{l})+4$. This is doable since $N^2(z_j^l)\cap N^2(z_{j'}^{l'}) \neq \emptyset$. Thus exists a net point $y$ in the sequence $z_{j'}^{1},\dots,z_{j'}^{l'},y_1,\dots,y_{i^*}$ of level $\level(z_{j}^{l})+3$. It should be noted that Lemma \ref{lem:available} guarantees the availability of all the net points within the sequence $z_{j'}^{1},\dots,z_{j'}^{l'},y_1,\dots,y_{i^*}$, including $y$, when the algorithm selects $z_{j'}^1$.

We have $\dist(z_{j}^{l},y_{i^*}) \leq 2\cdot 2^{-\level(z_{j}^{l})} + 2 \cdot 2^{-(\level(z_{j}^{l})+4)} \leq 16 \cdot 2^{-\level(y)} + 2^{-\level(y)} = 17 \cdot 2^{-\level(y)} $ and by Lemma \ref{lem:8diam} $\dist(y_{i^*},y) \leq 8 \cdot 2^{-\level(y)}$. On the other hand, a simple corollary of Lemma \ref{lem:8diam} is that $\dist(c_j, z_{j}^{l}) \leq 8\cdot 2^{-\level(z_{j}^{l})} = 64 \cdot  2^{-\level(y)}$. Using the triangle inequality, we get:

\begin{align*}
 \dist(c_j,y) &\leq  \dist(c_j, z_{j}^{l}) + \dist(z_{j}^{l},y_{i^*}) + \dist(y_{i^*},y)\\
&\leq 64 \cdot  2^{-\level(y)} + 17 \cdot 2^{-\level(y)} + 8 \cdot 2^{-\level(y)}\\
&= 89 \cdot 2^{-\level(y)}
\end{align*}

Therefore, $y \in N^{89}(z_{c_j})$ and this leads to a contradiction since $y$ became unavailable when $c_j$ was selected and therefore is not available when $z_{j'}^1$ is picked because $j>j'$. This conclude the proof of the first point.

We turn to the second point. Applying Lemma \ref{lem:8diam}, we get $\dist(z_{j'}^{l'}, z_{j'}^{1}) \leq 8 \cdot 2^{-\level(z_{j'}^{1})} \leq 0.5 \cdot 2^{-\level(z_j^{l})}$. On the other hand $N^2(z_j^l)\cap N^2(z_{j'}^{l'}) \neq \emptyset$ gives $\dist(z_{j}^{l},z_{j'}^{l'}) \leq 2 \cdot 2^{-\level(z_j^{l})} + 2\cdot 2^{-\level(z_j^{l}+4)} \leq 2.5 \cdot  2^{-\level(z_j^{l})}$. Using the triangle inequality, we have 

\begin{align*}
    \dist(z_{j}^{l},z_{j'}^{1}) &\leq \dist(z_{j}^{l},z_{j'}^{l'}) + \dist(z_{j'}^{l'}, z_{j'}^{1})\\
    &\leq 2.5 \cdot 2^{-\level(z_j^{l})} + 0.5\cdot 2^{-\level(z_j^{l})}\\
    & = 3 \cdot 2^{-\level(z_j^{l})}
\end{align*}

And $z_{j'}^{1}$ is in the $3$-neighborhood of $z_{j}^{l}$. Now let $w$ be a net point of level $\level(z_{j}^{l})+1$ such $z^1_{j'} \in N^{0.5}(w)$. Such a point exist by the covering property of nets. We have $N^1(z_{j'}^1) \subset N^1(w)$, hence $\val(w) \geq 8 \cdot \val(z_{j'}^1)$. Moreover $w$ is in $N^4(z_{j}^{l})$ and therefore $w$ is a child of $z_{j}^{l}$. The algorithm picked $z_{j}^{l+1}$ over $w$, so $\nval(z_{j}^{l+1}) \geq \nval(w)$. We know that at the end of the algorithm there exists an available net point with value greater than $2\cdot \thresh$. By the $\thresh$-threshold condition, the noisy value of this net point is at least $\thresh$. The algorithm picked $z_{j'}^1$ over this net points, so $\nval(z_{j'}^1) \geq \thresh$ and 

\begin{align*}
    \nval(z_{j}^{l+1}) &\geq \nval(w) \geq \val(w)/2 \geq 4 \val(z_{j'}^1) \geq 2  \nval(z_{j'}^1).\qedhere
\end{align*}
\end{proof}

We can now turn to the proof of \Cref{lem:pruning}, that we recall here for convenience:
\pruning*
\begin{proof}
We describe a procedure to compute the $l_j$'s. Start with $l_j = 1$ for all $j$, and note that with this choice the first condition is verified by design of the algorithm: when $z_j^1$ is picked (on line 5 of Algorithm~\ref{alg:mp}), it maximizes the noisy value among the available net points. In particular, $z_j^1$ has value larger than any cell that is still available at the end of the algorithm.
Then, to enforce the second condition, we proceed as follows: as long as we can find $j<j'$ and $l\geq l_j$, $l' \geq l_{j'}$ such that $N^2(z_j^{l}) \cap N^2(z_{j'}^{l'}) \neq \emptyset$, we update $l_j = l_j + 1$. 

The first item of \Cref{lem:pruningOneStep} guarantees that this procedure is well defined, namely that $z_j^{l}$ is not the last net point of the sequence and that $z_j^{l_j+1}$ does indeed exist. 
Furthermore, the second item of that Lemma ensures that $\nval(z_j^{l_j+1}) \geq \nval(z_{j'}^1)$, which itself is greater than $\nval(z)$ for any net point $z$ that is available at the end of Algorithm~\ref{alg:mp}. Therefore, the first condition remains satisfied after each update.
At each step, one of the $l_j$ get incremented, so this procedure must terminates because by the maximum level is $\lceil \log(n) \rceil$: when it ends, both conditions are satisfied, which concludes the proof. 
\end{proof}

\section{Missing Proofs for the $\eps$-DP Algorithm in the Low Dimensional Setting (Section~\ref{sec:privacyLD})}
\label{app:privacyLD}
Recall that by \emph{low-dimensional setting} we mean that we present an algorithm whose additive error is exponential in the dimension $d$, i.e., whose additive error is only small for small $d$.
\subsection{A Differentially Private Constant-Factor Approximation Algorithm}
In this subsection, we analyze \Cref{alg:makePriv}, and show \Cref{thm:linf}.

\linf*

To simplify the equations, we will note in the following
$\calE(\eps, \beta, n) = \frac{\log(n)^2\log(T)}{\eps} \cdot \lpar \log(1/\beta) + \sqrt{\log T}\rpar$.
As a direct consequence of \Cref{lem:dpCountNet}, we have the following fact: 
\begin{fact}\label{lem:propMakePriv}
    Algorithm \ref{alg:makePriv}  $\makePrivLowDim(P, \eps)$ satisfies $\eps$-differential privacy. 
    Moreover, it holds at each time step $t$, with probability $1-\beta$, that for all net point $z$,
    $\big|c(z,t) - |N^1(z) \cap P_t|\big| \leq d 2^{2d} \calE(\eps, \beta, n).$
\end{fact}
\begin{proof}
This directly stems from \Cref{lem:dpCountNet}, applied with $G = B(0, 1)$ and $G_j$ be the 1-neighborhoods of all net points. We have $b = 2^{2d} \log (n)$ and $m = O\left(n^{d}\right)$ as guaranteed by \Cref{lem:sizeNei}.
\end{proof}
Using this fact, we can show that the output of \Cref{alg:makePriv} satisfies the conditions of \Cref{thm:mp}:
\begin{lemma}\label{lem:estimation}
Let  $\thresh \geq 3 R$, where $R = d\cdot 2^{2d} \calE(\eps, \beta, n)$ is the upper bound from \Cref{lem:propMakePriv}. 
    For each time step $t$, it  holds with probability $1-\beta$ that, for all net point $z$: 
    \begin{itemize}
        \item if  $v(z,t) \geq \thresh$, then $v(z, t) \in [\val(z,t)/2, 2\val(z,t)]$,
        \item otherwise, if $v(z,t) <  \thresh$, then $\val(z,t) \leq 2 \cdot \thresh$. 
    \end{itemize}
\end{lemma}
\begin{proof}
Fix a time step $t$.    
By \Cref{lem:propMakePriv}, it holds with probability $1-\beta$ that  for all $z$, 
    $\big|c(z,t) - | N^1(z) \cap P_t|\big|  \leq R$. 

    Since $\val(z,t) = 2^{-2\ell}c(z, t)$ and $2^{-2\ell} \leq 1$,
     this implies immediately that $|v(z,t) - \val(z,t)| \leq   R$.
    The condition on $\thresh$ implies therefore that $|v(z,t) - \val(z,t)| \leq \thresh / 3$.
    
    First, in the case where $v(z,t) \geq  \thresh$, then $\val(z,t) \geq \frac{2\thresh}{3}$. Therefore, $v(z,t) \geq \val(z,t) - \thresh/3 \geq \val(z,t) - \frac{\val(z,t)}{2}$. Similarly, $v(z,t) \leq \val(z,t) + \thresh/3 \leq \frac{3\val(z,t)}{2}$. This concludes the first part:
    it holds that $v(z,t) \in [\frac{\val(z,t)}{2}, \frac{3\val(z,t)}{2}]$.
     
    In the case where $v(z,t) < \thresh$, note that 
        $\val(z,t) \leq v(z,t) +  \thresh / 3$:  
    therefore, if $v(z,t) \leq  \thresh$, then $\val(z,t) \leq 2  \cdot \thresh$.
\end{proof}

Equipped with those results, we can conclude the proof of our constant-factor approximation: essentially, we proved that the values computed with $\makePrivLowDim$ are private and that they satisfy the assumption of \Cref{thm:mp}. Thus \Cref{thm:mp} shows that an algorithm with the desired approximation guarantees exists.

\begin{proof}[Proof of \Cref{thm:linf}]
    We set the target probability $\beta = 0.99$.
    The algorithm first uses $\makePrivLowDim(P, \eps)$ to compute $v(A,t)$, and then $\mettuP(v,  k)$ to compute a solution for $k$-means.
    By \Cref{lem:propMakePriv}, the values $v(A,t)$ are private: by post-processing, the whole algorithm is therefore $\eps$-DP.
    For the approximation guarantee, \Cref{lem:estimation} ensures that the values  $v$ satisfy the assumption of \Cref{thm:mp} with $\thresh = d \cdot 2^{2d} \calE(\eps,\beta,n)$, at each step with probability $1-\beta$. Theorem~\ref{thm:mp} 
    therefore shows that the solution computed at time $t$ is an $(O(1), O(k) \cdot \thresh)$-approximation to $k$-means.
    {\setlength{\emergencystretch}{2.5em}\par}

    With $\beta = 0.99$, the additive error is therefore of the order $k d 2^{2d} \calE(\eps,0.99,n) = k d 2^{2d} \frac{\log(n)^2\log(T)^{3/2}}{\eps}$.\footnote{Note that we could have obtained an additive error $\frac{\log(n)\log(T)(\sqrt{\log T} + \log(n))}{\eps}$ instead, but we opted for simplicity, as in later applications for high dimensions this will be anyway dominated by applications of \Cref{lem:dpCountClusters}.}
\end{proof}

\subsection{Improving the approximation ratio from $O(1)$ to $(1+\alpha)w^*$-}\label{sec:optLD}
In this section we show how to improve the approximation ratio of our $\eps$-differentially private algorithm from constant to $(1 +\alpha) w^*$,
for any $0 < \alpha \le 1/4$, where $w^*$ is the approximation ratio of the best non-differentially private, static $k$-means clustering algorithm.

Given the mechanism $\calM'$ from \Cref{lem:dpCountClusters} with $ f \equiv 1$, our implementation of the algorithm described in \Cref{sec:privacyLD} in continual observation setting follows the ideas described in \Cref{sec:privacyLD}.

\begin{algorithm}
\caption{$\algLowDim(P, \eps)$, given a non-private algorithm $\calA$}
\label{alg:optLowDim}
\begin{algorithmic}[1]
\State{Let $v(\cdot, \cdot), c(\cdot,\cdot) = \makePrivLowDim(P, \eps)$}
\For{Each time $t$}
\State{Let $c_1^t, ..., c_{k'}^t$ be the centers computed by $\mettuP(v(\cdot, t), k')$}
\State{For $j$ in $1,..., k'$, let $n_j^t$ be the estimated number of points in cluster $C_j^t$ with mechanism $\calM'$ and failure probability $\beta = 0.01$}
\State{Let $\calS^t$ be the solution computed via $\calA$ on the dataset $P'_t$ containing $n_j^t$ copies of $c_j^t$, for all $j$.}
\EndFor
\State{\textbf{Output}: at time $t$, $\calS^t$.}
\end{algorithmic}
\end{algorithm}

\optLowDim*

\begin{proof}[Proof of \Cref{thm:optLowDim}]
This algorithm is private under continual observation: 
as shown in \Cref{thm:linf}, line 1 and 4 are private, and therefore robustness to post-processing ensures privacy of the solution computed line 5.

For the utility guarantee, we reproduce here the analysis of \cite{dpHD}. He showed that, for $k'= k\alpha^{-O(d)} \log(n/\alpha)$, then $\opt_{k'} \leq \alpha^2 \opt_k$, where $\opt_k$ is the optimal cost using $k$ centers and $\opt_{k'}$ with $k'$.

Therefore, the solution computed in our second step has cost $O(\alpha^2) \opt_k + O(k') \cdot \frac{ d 2^{2d}\log(n)^2\log(T)^{3/2}}{\eps}$, from \Cref{thm:linf}.
\Cref{lem:structClusters}  ensures the clusters computed with $\calM'$ yield a solution with additive error $O(k') \cdot \alpha^{-O(d)} \log(n)^3 \log(T)^{3/2} / \eps$, and therefore the total additive error is $k \cdot \alpha^{-O(d)} \log(n)^4 \log(T)^{3/2} / \eps$

In step 4, for each of the $ k\cdot 2^{O(d)} \log(n)$ cells, the estimated count from \Cref{lem:dpCountNet} is off by at most $\frac{\log(T)^{3/2}}{\eps} \cdot d \cdot \log(n/\beta)^2$.

This ensures that the dataset $P'_t$ defined line 5 is a good proxy for the dataset $P_t$ at time $t$: more precisely, for any candidate solution $S$, we have that 
\begin{align*}
    |\cost(P, S) - \cost(P', S)| &\leq \frac{\log(T)^{3/2}}{\eps} \cdot d \cdot \log(n/\beta)^2 + \sum_{j = 1}^{k'} \sum_{p \in {C'}^t_j} |\cost(p, S) - \cost(c_j^t, S)|\\
    &\leq  \frac{\log(T)^{3/2}}{\eps} \cdot d \cdot \log(n/\beta)^2 + \sum_{j = 1}^{k'} \sum_{p \in {C'}^t_j} \alpha \cost(p, S) + \frac{1+\alpha}{\alpha} \cost(p, c_j^t)\\
    &\leq \alpha \cost(p, S) + \frac{1+\alpha}{\alpha} \lpar O(\alpha^2) \opt_k + k \cdot \alpha^{-O(d)} \frac{\log(n)^4 \log(T)^{3/2}}{\eps} \rpar\\
    &\leq  \alpha \cost(p, S) + O(\alpha) \opt_k + k \cdot \alpha^{-O(d)} \frac{\log(n)^4 \log(T)^{3/2}}{\eps}
\end{align*}

Therefore, computing a $w^*$-approximation on this modified instance yields a solution with cost at most 
$(1+\alpha)w^* \opt_k + \frac{k\cdot 2^{O_\alpha(d)} \log(n)^4 \cdot\log(T)^{3/2}}{\eps}$, which concludes the theorem.
\end{proof}

\section{Missing Proofs for the $\eps$-DP Algorithm in the High Dimensional Setting (Section~\ref{sec:privacyHD})}\label{ap:hd}
Recall that by \emph{high-dimensional setting} we mean that we target an algorithm whose additive error is only polynomial in the dimension $d$.

\makePrivHD*
\begin{proof}
    The privacy guarantee follows from composition: computing $\calC^t$ in line 6 is $\eps/(d+3)$-DP, the computation of $\lbra n_j^t \rbra_{j, t}$ is $\eps/(d+3)$-DP, similarly to the computation of $\sumNorm(j,t)$ and the one of each $\sumEst(j,t)_i$ is $\eps/(d+3)$-DP. Therefore, the algorithm is $\eps$-DP. We note that the clusters ${C'}_j^t$ are not released by the algorithm, and therefore leak no privacy.

    The probability of success of the projection is $0.99$. In that case, with high probability all points $\sqrt{\frac{d}{\hat d}} \pi(p)$ are in the ball $B_{\hat d}(0, \log (n))$, the projection of line $3$ does not change the location of points, as shown in~\cite{dpHD}.  This step is nonetheless necessary to guarantee that the diameter (and therefore, the privacy loss) is bounded in case the dimension reduction fails.
    
    In the case the projection succeeds, the first bullet in the lemma statement is as well a direct consequence of \Cref{thm:optLowDim} and \Cref{lem:dpCountClusters}: the additive error is dominated by the one from \Cref{thm:optLowDim} and is  $O\lpar k \alpha^{-O(\hat d)} \log(n)^4 \cdot \frac{\log(T)^{3/2}}{\eps/(d+2)} \rpar = d \calE(\eps, 0.99, n) \log(n)^2 k^{O(1)}$, which concludes the first bullet of the lemma.
    
    To show the second bullet, we use again \Cref{lem:dpCountClusters} for the estimation of $\sumEst(j, t)$: for each coordinate $i$, \Cref{lem:dpCountClusters} is applied with function $f_i(p') = p_i$ (and it holds that $f_i(p) \leq 1$ since $p \in B_d(0, 1)$), and success probability $1-\beta/(d+3)$:
    
    \begin{align*}
        \logor \sumEst(j, t) - \sum_{p: p' \in {C'}_j^t} p\rnor_2^2 &= \sum_{i=1}^d \lpar \sumEst(j, t)_i - \sum_{p \in C^t_j} p_i\rpar^2\\
        &\leq \sum_{i=1}^d \lpar  k 2^{O(\hat d)} \log(n)^3 \frac{\log(T)(\sqrt{\log T} + \log(d/\beta))}{\eps/d}  \rpar^2\\
        &\leq d \cdot \lpar  k^{O(1)} \log(n)^3  \frac{\log(T)(\sqrt{\log T} + \log(d/\beta))}{\eps/d}   \rpar^2\\
        &=   d^{3} \cdot \lpar \log(d) \calE(\eps, \beta, n) \log(n) k^{O(1)} \rpar^2.
    \end{align*}

    Taking the square root concludes. Finally, we bound $\sumNorm(j, t)$ by a direct application of \Cref{lem:dpCountClusters}, using the function $f(p') = \|p\|_2^2$ (for which it holds again that $f(p') \leq 1$). With probability $1-\beta/(d+3)$ it holds that:
    \begin{align*}
        \left| \sum_{x \in C_j} \|x\|_2^2 - \sumNorm(j,t)\right| &\leq k \alpha^{-O(\hat d)} \log(n)^3 \frac{\log(T)(\sqrt{\log T} + \log(d/\beta))}{\eps/d} \\
        &\leq d k^{O(1)} \cdot \calE(\eps, \beta, n) \log n.
    \end{align*}
    This concludes the proof of the lemma.
\end{proof}

\hdkmeans*
\begin{proof}
Note that the call to Algorithm \ref{alg:hd}  fixes the target failure probability $\beta$ to $1/10$.
The privacy guarantee follows from \Cref{lem:makePrivHD}.
For simplicity of notation, we fix a time $t$ and remove mentions of the time in our analysis.

We need to show that the centers $\bc_1, ..., \bc_k$ computed by \Cref{alg:hd} induce a solution in $\R^d$ that has cost comparable to the one induced by the clustering ${C'}_1, ..., {C'}_k$ computed in \Cref{alg:hd} in line $6$.

\textbf{The implicit clustering is a good approximation.}
Note that $C'_j$ is a cluster for the dataset $P'$: we let $C_j$ be the cluster of $P$ formed of points $p$ in $\R^d$ with $p'  \in C'_j$. 
We first show that assigning each of $C_1,...,C_k$ to its mean $\mu(C_j)=\frac{\sum_{p \in C_j} p}{|C_j|}$ is a good approximation to the $k$-means clustering for $P$.
    
The dimension reduction $\pi$ succeeds with probability $0.99$. In that case, since the cost of any partition is preserved up to a factor $(1\pm \alpha)$, then any $(a,b)$-approximation for $\pi(P)$ is an $((1+ \alpha)a, b)$-approximation for $P$.
Now, \Cref{thm:optLowDim} shows that, in $\R^{\hat d}$, the clusters $C'_1, ..., C'_k$ are an $\lpar (1+\alpha)w^*, k^{O(1)} \calE(\eps, 1/10, n) \log(n)^2 \rpar$-approximate solution for $\pi(P)$ with probability $9/10$.
    {\setlength{\emergencystretch}{2.5em}\par}

    Therefore, with a union-bound, it happens with probability at least $1-2\cdot 1/10  = 4/5$ (when both dimension reduction and the previous argument are successful) that clustering each $C_j$ to its mean yields an $\lpar (1+\alpha)w^*, k^{O(1)} \calE(\eps, 1/10, n) \log(n)^2 \rpar$-approximate solution for $P$.    
    {\setlength{\emergencystretch}{2.5em}\par}

\textbf{Bounding the error in each cluster.}
To conclude the proof of the theorem it is  enough to show that for all $j \in \lbra 1, ..., k\rbra$, $\cost(C_j, \bc_j) \leq \cost(C_j, \mu(C_j)) + k^{O(1)} \cdot d^2 \cdot \frac{\log(n)^4\log(T)^{3/2}}{\eps}.$

    To do so, we first prove that $\bc_j$  is close to $\mu(C_j) = \frac{\sum_{p \in C_j} p}{|C_j|}$ and then bound the desired cost difference: For this, we use that both the numerator and denominator are well approximated. Indeed, \Cref{lem:makePrivHD} ensures that for $\beta = 1/10$:
     \begin{align}
     \label{eq:deno}
         |n_j - |C_j|| &\leq \err,\\
     \label{eq:num}
         \logor \sumEst(j) - \sum_{p \in C_j} p\rnor_\infty &\leq  \err \cdot d^{1/2} \log(d),
     \end{align} 
where $\err = k^{O(1)} \calE(\eps, 1/10, n) \cdot d \log(n) 
= k^{O(1)} d \cdot \log(T)^{3/2} \log (n)^3/\eps$

    If a cluster contains fewer than $2\cdot \err$ points, then regardless of the position of $\bc_j$ its $1$-means cost is upper-bounded by $2\err$ -- as the diameter of $P$ is $1$ in $\R^d$-- and this is accounted for by the additive error.
    Otherwise, we can show that $\bc_j$ is very close to $\mu(C_j)$, using \Cref{eq:num}, \Cref{eq:deno}, and the fact that $|C_j| \ge 2\err$ for all $j$ as follows:

    \begin{eqnarray}
    \notag
        \logor \frac{\sumEst(j)}{n_j} - \mu(C_j)\rnor_2 
        &=& \logor  \lpar \frac{\sumEst(j)}{n_j} - \frac{\sumEst(j)}{|C_j|}\rpar + \lpar \frac{\sumEst(j)}{|C_j|} - \frac{|C_j|\mu(C_j)}{|C_j|}\rpar \rnor_2\\
        \notag
        \text{(\ref{eq:deno})}&\leq& \logor \sumEst(j)\rnor_2 \lpar \frac{1}{|C_j| - \err} - \frac{1}{|C_j|}\rpar + \frac{\logor \sumEst(j) - |C_j|\mu(C_j)\rnor_2}{|C_j|}\\
    \end{eqnarray}
    Using Eq.~\ref{eq:num}, and $(1-x)^{-1} \leq 1 + 2x$ for $x \leq 1/2$, this is upper-bounded by:
    \begin{eqnarray}
         \logor \frac{\sumEst(j)}{n_j} - \mu(C_j)\rnor_2 &\leq& \logor \sumEst(j)\rnor_2 \lpar \frac {1} {|C_j|}   + \frac{2\err}{|C_j|^2} - \frac{1}{|C_j|}\rpar + \frac{d^{1/2} \log (d)  \cdot \err}{|C_j|}\\
        \notag
        &\leq& \frac{|C_j| \logor \mu(C_j)\rnor_2 + \logor \sumEst(j) - |C_j|\mu(C_j)\rnor_2}{|C_j|}\cdot \frac{2\err}{|C_j|}  + \frac{d^{1/2} \log (d) \cdot \err}{|C_j|}\\
        \notag
        \text{(Eq.~\ref{eq:num})}&\leq &\lpar \logor \mu(C_j)\rnor_2 + \frac{d^{1/2} \log (d)  \cdot \err}{|C_j|}\rpar \frac{2\err}{|C_j|} + \frac{d^{1/2} \log (d)  \cdot \err}{|C_j|}\\
        \notag
        C_j \subset B_d(0, 1): &\leq& \lpar 1 + \frac{d^{1/2} \log (d)\cdot \err}{|C_j|}\rpar \frac{2\err}{|C_j|} + \frac{d^{1/2} \log (d)  \cdot \err}{|C_j|}\\
        \label{eq:estMean}
        &\leq& \frac{O(d^{1/2} \log (d) \cdot \err)}{|C_j|},
    \end{eqnarray}
    where the last line follows from $|C_j| \geq 2\err$.

    Finally, we get, using  the folklore property that for any set $X$ and point $y$, $\cost(X, y) = \cost(X, \mu(X)) + |X| \logor \mu(X)-y\rnor_2 ^2$ (see e.g. \cite{BecchettiBC0S19}):
    \begin{align*}
        \cost(C_j, \bc_j) &= \cost(C_j, \mu(C_j)) + |C_j| \cdot \logor \bc_j - \mu(C_j)\rnor_2^2\\
        &\leq \cost(C_j, \mu(C_j)) +\frac{O(\lpar d^{1/2} \log (d) \cdot \err \rpar)^2}{|C_j|} \\
        &\leq \cost(C_j, \mu(C_j)) + k^{O(1)} \cdot d^2 \log(d)^2 \cdot \frac{\log(n)^3 \log(T)^{3/2} }{\eps},
    \end{align*}
    where the last line uses again $|C_j| \geq 2\err$.
    Summing over all clusters concludes the proof: the additive error is dominated by the additive error incurred by the algorithm of \Cref{thm:optLowDim}, and is $k^{O(1)} \cdot d^2 \cdot \frac{\log(n)^4\log(T)^{3/2}}{\eps}$.
\end{proof}

\subsection{Boosting the Success Probabilities (\Cref{sec:boostProba})} \label{ap:boostProba}
To boost the probabilities so that the algorithm outputs solutions that are good at all steps simultaneously, we will run several copies and at each time output the solution with best cost. For that, it is necessary to measure the cost of a cluster.
Our first lemma formalizes how to do this without having access to its center -- which will be helpful to not lose additional privacy.
\begin{restatable}{lemma}{costVariance}\label{lem:costVariance}
    For any set $X$, $\cost(X, \mu(X)) =\sum_{x\in X} \|x\|_2^2 - \frac{\logor \sum_{y \in X} x\rnor_2^2}{|X|}$.
\end{restatable}
\begin{proof}
    Using properties of squared distances, we have the following:
    \begin{align*}
        \cost(X, \mu(X)) &= \sum_{x\in X} \|x - \mu(X)\|_2^2 
        = \sum_{x\in X} \sum_{i = 1}^d |x_i - \mu(X)_i|^2\\
        &= \sum_{x\in X} \sum_{i = 1}^d x_i^2 + \mu(X)_i^2 - 2 x_i \mu(X)_i\\
        &= \sum_{x\in X} \|x\|_2^2 + \|\mu(X)\|_2^2  - 2 \sum_{i=1}^d x_i \mu(X)_i.
    \end{align*}
    Here, we note that $\mu(X) = \frac{1}{|X|}\sum_{y \in X} y$: therefore, $ \sum_{x\in X} \|\mu(X)\|_2^2 = |X| \cdot \|\mu(X)\|_2^2 
 =  \frac{\logor \sum_{y \in X} y \rnor_2 ^2}{|X|}$. Furthermore, $\mu(X)_i = \frac{1}{|X|}\sum_{y \in X} y_i$, and so:
    \begin{align*}
        \sum_{x\in X} \sum_{i=1}^d x_i \mu(X)_i &= \sum_{i=1}^d \lpar \frac{1}{|X|}\sum_{y \in X} y_i\rpar \lpar \sum_{x\in X} x_i \rpar \\
        &= \frac{1}{|X|}\sum_{i=1}^d \lpar \sum_{x \in X} x_i\rpar^2\\
        &=  \frac{\logor \sum_{x \in X} x\rnor_2^2}{|X|}.
    \end{align*}

    Combining those equations concludes the lemma:
        \begin{align*}
        \cost(X, \mu(X)) 
        &= \sum_{x\in X} \|x\|_2^2 + \|\mu(X)\|_2^2  - 2 \sum_{i=1}^d x_i \mu(X)_i\\
        &= \sum_{x\in X} \|x\|_2^2 - \frac{\logor \sum_{y \in X} x\rnor_2^2}{|X|}.\qedhere
    \end{align*}
\end{proof}

\costEst*
\begin{proof}
Fix any time $t$ and cluster $C'_j$. We note that, since $C'_j$ is a cluster of points in dimension $\hat d = O(\log(k))$,  all calls to \Cref{lem:dpCountClusters} incur with probability $1-\beta/T$ an error $k^{O(1)} \log(n)^3 \cdot \frac{\log(T)}{\eps} \lpar \sqrt{\log T} + \log(T/\beta)\rpar$.

From \Cref{lem:costVariance}, it holds that 
    \begin{equation}
        \cost(C_j, \mu(C_j)) = \sum_{x \in C_j} \|x\|_2^2 - \frac{\logor  \sum_{x \in C_j} x\rnor_2 ^2}{ |C_j|}. 
        \label{eq:costVar}
    \end{equation}
    We bound the difference  with $\sumNorm(j) - \frac{\logor \sumEst(j)\rnor_2 ^2}{n_j}$  term by term, and let $\err' := d k^{O(1)} \cdot \calE(\eps, \beta / T, n) \log n = d k^{O(1)} \frac{\log(n)^3 \log(T)}{\eps} \log(T/\beta)$ for simplicity.

Essentially, \Cref{lem:makePrivHD} ensures that, for any time step $t$, it holds with probability at least $1-\beta/T$ that all $d+3$ values returned by \Cref{alg:hd} are within their respective error bounds, i.e., $n_j$ is a good estimate of $|C_j|$, $\sumNorm(j)$ is a good estimate of $\sum_{p \in C_j}\logor p \rnor_2$, and $\sumEst(j)$ is a good estimate of $\sum_{p \in C_j}p$ 
(as the calls to the algorithm of \Cref{lem:dpCountClusters} are with parameter $\beta/(dT+3T)$, and there are $d+3$ calls in \Cref{alg:hd}). By a union-bound, those properties hold with probability $1-\beta$ for all time steps simultaneously.
    
    \emph{Case 1:} First, in the case where $|C_j| \leq 2\err'$, then 
    \Cref{lem:makePrivHD} and the fact that all points belong to the unit ball, i.e., $\|p\|_2 \le 1$, ensure that
    $\sumEst(j) \leq \|\sum_{p : p' \in C_j} p\|_2 + \| \sumEst(j) - \sum_{p : p' \in C_j} p \ \|_2  \leq \sum_{p : p' \in C_j} \| p\|_2 + \| \sumEst(j) - \sum_{p : p' \in C_j} p\|_2 \leq 2\err' + O( d^{1/2} \log (d) \err') \leq   O(d^{1/2} \log(d) \err')$.
    We can bound $\sumNorm$ similarly: \Cref{lem:makePrivHD} shows that $\sumNorm(j,t) \leq \sum_{x \in C_j^t} \| x \|_2^2 + \err' \leq |C_j^t| + \err' \leq 3\err'$.
    Therefore, with probability at least $1 - \beta$, 
    $\left| \cost(C_j, \mu(C_j)) - \lpar \sumNorm(j,t) - \frac{\logor \sumEst(j)\rnor_2 ^2}{n_j}\rpar\right| \leq O(d^{1/2} \log(d) \err')$, as all terms are at most $O(d^{1/2} \log(d) \err')$, which concludes this case.

    \emph{Case 2:} In the other case, when  $|C_j| \geq 2 \err'$, then, \Cref{lem:makePrivHD} with failure probability $\beta/T$ ensures that, with probability $1-\beta/T$ it holds that for all cluster $C_j$,
        \begin{align*}
        \left| \sum_{x \in C_j} \|x\|_2^2 - \sumNorm(j,t)\right| 
        &\leq d k^{O(1)} \cdot \calE(\eps, \beta / T, n) \log n = \err'.
    \end{align*}

   A union-bound over all time steps $t$ ensures that this inequality holds with probability $1-\beta$  for all time steps simultaneously.

   For the second term of \Cref{eq:costVar}, we use exactly the same reasoning as in Equation~(\ref{eq:estMean}) from the proof of \Cref{lem:hdkmeans}: it holds with probability $1-\beta$, for any time $t$:
   \begin{align*}
           \logor \frac{\sumEst(j)}{ |n_j|} - \mu(C_j)\rnor_2 \leq \frac{\err' d^{1/2} \log(d) }{|C_j|}.
   \end{align*}
   We formally prove this claim in \cref{claim:distMean}.
   Furthermore, we know that $\left|n_j - |C_j| \right| \leq \err'$ and $\|\mu(C_j)\|_2^2 \le 1$ for all $1 \le j \le k$. Therefore, the following holds:
   \begin{align*}
       \frac{\logor  \sumEst(j)\rnor_2 ^2}{ n_j} 
       &\geq  n_j \cdot \lpar \|\mu(C_j)\|_2 - \logor \frac{  \sumEst(j)}{n_j} - \mu(C_j)\rnor_2 \rpar^2\\
       &\geq (|C_j| - \err') \cdot  \lpar \|\mu(C_j)\|_2^2  - 2 \|\mu(C_j)\|_2 \frac{\err' d^{1/2} \log(d) }{|C_j|}\rpar\\
       &\geq |C_j| \cdot \|\mu(C_j)\|_2 - \err' \|\mu(C_j)\|_2^2  - 2\|\mu(C_j)\|_2 \err' d^{1/2} \log(d)  \\
       &\geq |C_j| \cdot \|\mu(C_j)\|_2 - \err'   - 2 \err' d^{1/2} \log(d) \\
       &\geq  |C_j| \cdot \|\mu(C_j)\|_2 -   k^{O(1)} \frac{d^{3/2} \log(d) \log(n)^3 \log(T)}{\eps} \log(T/\beta).
   \end{align*}

   One can similarly show that 
   \[\frac{\logor  \sumEst(j)\rnor_2 ^2}{n_j}  \leq |C_j| \cdot \|\mu(C_j)\|_2 + k^{O(1)} \cdot \frac{d^{3/2} \log(d) \cdot \log(n)^3 \log(T) \log(T/\beta)}{\eps}.\]

   Therefore, we can conclude:
   \[\left| \cost(C_j, \mu(j)) - \lpar \sumNorm(j) - \frac{\logor \sumEst(C_j)\rnor_2 ^2}{n_j}\rpar\right| \leq k^{O(1)} \cdot \frac{d^{3/2} \log(d) \cdot \log(n)^3 \log(T) \log(T/\beta)}{\eps}\]
\end{proof}

\begin{fact}\label{claim:distMean}
Let $\err' = k^{O(1)} \frac{d \log(n)^3 \log(T)}{\eps} \log(T/\beta)$, and $C_j$ be a cluster with $|C_j| \geq 2 \err'$.
With probability $1-\beta$, it holds for any time $t$ that:
   \begin{align*}
           \logor \frac{\sumEst(j)}{ |n_j|} - \mu(C_j)\rnor_2 \leq \frac{\err' d^{1/2}  \log(d)}{|C_j|}.
   \end{align*}
\end{fact}
\begin{proof}
This proof is identical to the one for showing Equation~(\ref{eq:estMean}), using the bound $|C_j| \geq \err'$ instead of $|C_j|\geq \err$. We repeat it for completeness.
     \Cref{lem:makePrivHD} ensures that with probability $1-\beta$:
     \begin{align}
     \label{eq:deno2}
         |n_j - |C_j|| &\leq \err',\\
     \label{eq:num2}
         \logor \sumEst(j) - \sum_{p \in C_j} p\rnor_\infty &\leq  \err' \cdot d^{1/2} \log(d),
     \end{align}

Combining triangle inequality and Eq.~\ref{eq:deno2}, we have:
    \begin{eqnarray*}
        \logor \frac{\sumEst(j)}{n_j} - \mu(C_j)\rnor_2 
        &=& \logor  \lpar \frac{\sumEst(j)}{n_j} - \frac{\sumEst(j)}{|C_j|}\rpar + \lpar \frac{\sumEst(j)}{|C_j|} - \frac{|C_j|\mu(C_j)}{|C_j|}\rpar \rnor_2\\
        \text{(\ref{eq:deno2})}&\leq& \logor \sumEst(j)\rnor_2 \lpar \frac{1}{|C_j| - \err'} - \frac{1}{|C_j|}\rpar + \frac{\logor \sumEst(j) - |C_j|\mu(C_j)\rnor_2}{|C_j|}\\
    \end{eqnarray*}
    Using Eq.~\ref{eq:num2}, and $(1-x)^{-1} \leq 1 + 2x$ for $x \leq 1/2$, this is upper-bounded by:
    \begin{eqnarray*}
         \logor \frac{\sumEst(j)}{n_j} - \mu(C_j)\rnor_2 &\leq& \logor \sumEst(j)\rnor_2 \lpar \frac {1} {|C_j|}   + \frac{2\err'}{|C_j|^2} - \frac{1}{|C_j|}\rpar + \frac{d^{1/2} \log (d)  \cdot \err'}{|C_j|}\\
        &\leq& \frac{|C_j| \logor \mu(C_j)\rnor_2 + \logor \sumEst(j) - |C_j|\mu(C_j)\rnor_2}{|C_j|}\cdot \frac{2\err'}{|C_j|}  + \frac{d^{1/2} \log (d) \cdot \err'}{|C_j|}\\
        \text{(Eq.~\ref{eq:num})}&\leq &\lpar \logor \mu(C_j)\rnor_2 + \frac{d^{1/2} \log (d)  \cdot \err'}{|C_j|}\rpar \frac{2\err'}{|C_j|} + \frac{d^{1/2} \log (d)  \cdot \err'}{|C_j|}\\
        C_j \subset B_d(0, 1): &\leq& \lpar 1 + \frac{d^{1/2} \log (d)\cdot \err'}{|C_j|}\rpar \frac{2\err'}{|C_j|} + \frac{d^{1/2} \log (d)  \cdot \err'}{|C_j|}\\
        &\leq& \frac{O(d^{1/2} \log (d) \cdot \err')}{|C_j|},
    \end{eqnarray*}
    where the last line follows from $|C_j| \geq 2\err'$.
\end{proof}

We know from \ref{lem:hdkmeans} that the solution at any time step is correct with probability $4/5$, and moreover \Cref{lem:costEst} estimates with probability $1-\beta$ the cost at any time step. Knowing $T$ in advance, one solution to boost probabilities would be to run $\lfloor \log(T) \rfloor$ independent copies of $\hdkm(P, \eps / \log(T), 1/5)$: at each time step at least one of those copies will be a good approximation with probability $1-1/5^{\lfloor \log(T) \rfloor} \approx 1-1/T$. Selecting the cheapest according to \Cref{lem:costEst} would therefore conclude that with constant probability the solution at each time $t$ is a good approximation. 

In the case where $T$ is unknown, one can instead run $\lfloor \log(T) \rfloor$ instances of  $\hdkm(P, \eps, 1/5)$, with parameters defined as follows: the instance started at time $t = 2^i$ has privacy parameter $\eps_i$. 
Starting an instance at time $2^i$ means that we create a stream that inserts one by one each point currently in the data, in the order they appeared, and run a copy of algorithm $\hdkm$ starting with this stream and continuing with the actual updates.

This whole algorithm is $O(\sum \eps_i)$-DP, using standard composition. Furthermore, the utility increases only by a factor $\max_{i \leq \log(T)} \frac{1}{\eps_i}$. 
We will use the fact that $s(\kappa):= \sum i^{-1-\kappa} = O(1/\kappa)$ and $\max_{i \leq \log(T)} {i^{1+\kappa}} = \log(T)^{1+\kappa}$:
thus, setting $\eps_i = \frac{\eps}{s(\kappa) i^{1+\kappa}}$ 
ensures that, with constant probability, the algorithm is $\eps$-differentially private and gives essentially an $\lpar (1+\alpha)w^*,  k^{O(1)}  \cdot d^2 \log(1/\beta) \log(n)^4 \log(T)^{3+\kappa}  \cdot {\eps^{-1}}\rpar$-approximation to $k$-means.  
The next algorithm and lemma make formal this proof sketch.

The formal algorithm we use to boost the probabilities is the following:
\begin{algorithm}[H]
\caption{BoostProbabilities$(P, \eps, \beta, \kappa)$}
\label{alg:boostProba}
\begin{algorithmic}[1]
\State let $s(\kappa) = \sum_{i \geq 1} i^{-1-\kappa}$, and $c_1, c_2$ be two large enough constants ($c_1 \geq 3/\log(5/2), c_2 \geq \log(\pi^2/6\beta) / \log(5/2)$)
\State{For $i = \lbra 1,..., c_1 \lfloor \log (t) \rfloor  + c_2\rbra$, run a copy of $\hdkm(P,  \eps /(s(\kappa) i^{1+\kappa}), 1/5)$.}
\State{Estimate the cost of each solution using \Cref{lem:costEst}.}
\State{Output centers corresponding to the solution with cheapest cost.}
\end{algorithmic}
\end{algorithm}

The next lemma is a more precise formulation of \Cref{thm:main} -- which directly follows by setting $\kappa = 0.001$.

\boostProba*
\begin{proof}
We show that algorithm~\ref{alg:boostProba} satisfies the lemma.
   The composition theorem and the properties of \Cref{alg:hd} shown in \Cref{lem:hdkmeans} ensure that the resulting algorithm is $\eps$-DP as     $ \sum_{i=1}^{c_1 \log (T) + c_2} \frac{\eps}{s(\kappa) i^{1+\kappa}} \leq \eps$ (by definition of $s(\kappa)$).
   
    Furthermore, each instance provides at time $t$, with probability $4/5$, an approximation with multiplicative factor unchanged compared to \Cref{lem:hdkmeans}, and additive error increased (due to the rescaling of $\eps$) by a factor $s(\kappa) i^{1+\kappa} \leq s(\kappa) (c_1\log (t) + c_2)^{1+\kappa} = O(\kappa) \cdot \lpar\log(T)^{1+\kappa} + \log(1/\beta)^{1+\kappa}\rpar$.

    Since all copies are independent, at least one succeeds with probability $1-(2/5)^{c_1 \log(T) + c_2} \geq 1 - \frac{6 \beta}{\pi^2 \cdot T^2}$. Namely, at least one solution is an approximation with multiplicative factor $(1+\alpha)w^*$ and additive at most (where we took the maximum for each exponent of the additive terms from \cref{lem:hdkmeans} and \cref{lem:costEst}):
    \begin{align*}
         & k^{O(1)} \cdot \frac{d^{2} \cdot \log(n)^4 \log(T) \log(T/\beta)}{\eps/(s(k) i^{1+\kappa})} \\
         = ~& k^{O(1)} \cdot \frac{d^{2}  \cdot \log(n)^4 \log(T)^{3+\kappa} \log(1/\beta)}{\eps}
    \end{align*}

    We conclude the proof with a union-bound over all $T$ steps: for this choice of $c_1, c_2$, with probability $1-\frac{\beta \cdot 6}{\pi^2}\cdot\sum_{t = 1}^T 1/t^2 \geq 1-\beta$, the solution is correct at any time $t$.
\end{proof}

\section{Missing Proof for Privately Counting in Clusters (\Cref{sec:dpCountClusters})}\label{ap:dpCountClusters}
In this section we show the proof of
\Cref{lem:structClusters}, which is the last missing piece to complete the proof of the crucial
Lemma~\ref{lem:dpCountClusters}.

To define the sets $G_1, ... G_m$, we make use of a 
\emph{hierarchical decomposition} $\calG = \lbra \calG_1,..., \calG_{\lceil \log (n)\rceil}\rbra$ with the following properties: 
\begin{enumerate}
    \item for all  $i = 0, ..., \lceil \log (n)\rceil$, $\calG_i$ is a partition of $B_d(0,1)$. Each part $A$ of $\calG_i$ is called a cell of level $i$, denoted $\level(A) = i$, and has diameter at most $2^{-i}$.
    \item The partition $\calG_{i+1}$ is a refinement of the partition $\calG_i$, namely every part of $\calG_{i+1}$ is strictly contained in one part of $\calG_i$. $\calG_0$ is merely $\lbra B_d(0, 1)\rbra$.
    \item for any point $p \in B_d(0, 1)$, any level $i$ and any $r \geq 1$, the ball $B_d(p, r 2^{-i})$ intersects at most $r^{O(d)}$ many cells of level $i$.
\end{enumerate}
For a cell $A$ of level $i+1$, we denote by $\parent(A)$  the unique cell of level $i$ containing $A$. 

The sets $G_1,..., G_m$ on which we will apply \Cref{lem:dpCountNet} are the cells at each level. Since they are $\log (n)$ levels, each being a partition, any point from $B_d(0, 1)$ appears in $b = \log (n)$ sets. 
Furthermore, $m = O\lpar n^{d}\rpar$ (by the third property with $r = 1/n$), and so \Cref{lem:dpCountNet} applies with additive error $O\lpar d \cdot \log(n)^2 \cdot \frac{\log(T)}{\eps} \cdot \lpar \sqrt{\log (T)} + \log(1/\beta)\rpar\rpar$. 

Therefore, in order to maintain $\sum_{p \in C_i^t} f(p)$, we only need to show how to express this sum with few $\sum_{p \in G_i \cap P_t} f(p)$ (where $P_t$ is the dataset at time $t$). For this, we use the hierarchical decomposition described above, and identify few cells that are far from any center, compared to their diameter. All points in such a cell are roughly equivalent for the given clustering, and therefore we can serve them all by the same center.  This idea is formalized in the next lemma:

\structClustersNet*
\begin{proof}
    Let $\ell = \lceil 10/\alpha \rceil$. For any cell $A$ of the decomposition, fix an arbitrary point $v_A \in A$. We call the $\ell$-neighborhood of $A$ all cells at the same level as $A$ that intersect with the ball $B\lpar v_A, \ell 2^{-\level(A)}\rpar$. Note that this includes $A$.
    Let $\calA$ be the set of cells $A$ such that their $\ell$-neighborhood does not contain a center of $\calC$, but the $\ell$-neighborhood of $\parent(A)$  does contain one. Add furthermore to $\calA$ the cells at level $\lceil \log (n) \rceil$ that contain a center in their $\ell$-neighborhood.

    We first bound the size of the set $\calA$ obtained. For this, we will count for each level how many cells have a center in their $\ell$-neighborhood, and how many children each cell has. 
    First, for a fixed level $i$ and center $c^*$, the decomposition ensures that there are at most $\ell^{O(d)}$ many cells intersecting $B_d(c^*,  2\ell 2^{-i})$. Note that any cell $A$ that contains $c^*$ in its $\ell$-neighborhood must intersect $B_d(c^*,  2\ell 2^{-i})$: indeed, if $B$ is the cell containing $c^*$ this means there is a point in $B$ at distance at most $\ell 2^{-i}$ of $v_A$. Since the diameter of $B$ is at most $2^{-i}$, $c^*$ is at distance at most $\ell 2^{-i} + 2^{-i}$ of $v_A$, and therefore $A$ intersects with the ball $B_d(c^*,  2\ell 2^{-i})$. 
    
    Now, using the third property of decomposition, any cell $A$ of level $i$ is the parent of at most $2^{O(d)}$ many cells. Indeed, all those cells are included in $A$: this means they are contained in a ball of radius $2^{-i}$, and since they are at level $i+1$, third property ensures that there are at most $2^{O(d)}$ many of them.
    
    Therefore, at any level, at most $k\cdot \ell^{O(d)} = k \cdot \alpha^{-O(d)}$  cells $A$ are added, as $N^\ell(\parent(A))$ must contain one center. Hence, $|\calA| \leq k \cdot \alpha^{-O(d)} \log(n)$.

    \textbf{$\calA$ is a partition of $B_d(0, 1)$.} Now, we show that $\calA$ is a partition of $B_d(0, 1)$. For this, we observe that if a cell $A$ contains a center in its $\ell$-neighborhood, then $\parent(A)$ also contains one.
    First, we argue that the cells in $\calA$ cover the whole unit ball. For a point $p$ in the ball, consider the sequence of parts containing it. There are two cases: Either the smallest part at level $\lceil \log (n) \rceil$ contains a center in its $\ell$-neighborhood, and, thus, it is added to $\calA$; or it does not, and then the largest part of the sequence that does not contain a center in its $\ell$-neighborhood is added to $\calA$. Note that such a largest part must exist, as the cell of level $0$ contains all centers.

    Next we argue that all  cells 
    in $\calA$ are disjoint. Consider two  intersecting cells $A$ and $B$ in $\calA$. 
    By property 2 of the decomposition, it must be that one is included in the other: assume wlog that $A \subset B$ (therefore $\level(A) > \level(B)$). 
    Note that if $A$ is on level $\log (n)$ and has a center in its $\ell$-neighborhood, then trivially $N^\ell(\parent(A))$ contains a center.
    Thus, for all $A \in \cal A$ it holds that $N^\ell(\parent(A))$ contains a center.
    But $\parent(A) \subseteq B$: therefore, $N^\ell(B)$ contains a center, which implies that $B \notin \calA$. 
    Thus, $A$ and $B$ cannot be both in $\calA$, which concludes the proof that $\calA$ is a partition of $B_d(0, 1)$. 

    \textbf{Assigning cells to center.} We can now define $c$ as follows: for each cell $A \in \calA$, we define $a(A) = \argmin_{i} \dist(v_A, c_i)$ to be the closest point from $\calC$ to $v_A$.

    We first define an assignment for a cell  $A \in \calA$ at level $\log (n)$. In this case, triangle inequality ensures that paying the diameter of $A$ for each point allows to serve all points in $A$ with the same center. Formally, let $a(A)$ be the closest center to $v_A$. For any $p\in A$, we have using \Cref{lem:weaktri}:
    \begin{align*}
        \dist(p, a(A))^2 &\leq (\dist(p, v_A) + \dist(v_A, a(A)))^2 \leq (2\dist(p, v_A) + \dist(p, \calC))^2\\
        &\leq (1+\alpha/2)\dist(p, \calC)^2 + (1+\frac{2}{\alpha})\cdot 4\dist(p, v_A)^2 \\
        &\leq (1+\alpha/2)\dist(p, \calC)^2 +\frac{9}{\alpha}\cdot \frac{1}{n}.
    \end{align*}
    Therefore, summing all points in cells at level $\log (n)$ gives an additive error at most $\frac{9}{\alpha}$.

    Now, fix a cell $A \in \calA$ at level $> \log (n)$, and a point $p \in A$. Since, by construction of $\calA$, the $\ell$-neighborhood of $A$ contains at most one center from $\calC$, it holds that:
    \begin{itemize}
        \item either $p$ is already assigned to $a(A)$, and then $\dist(p, a(A)) = \dist(p, \calC)$,
        \item or the center serving $p$ in $\calC$ is outside of the $\ell$-neighborhood of $A$: in particular, $\dist(p, \calC) \geq \ell \cdot \dist(p, v_A)$. In that case, we use the modified triangle inequality from \Cref{lem:weaktri}. This yields similarly as above:
    \begin{align*}
        \dist(p, a(A))^2 &\leq (\dist(p, v_A) + \dist(v_A, a(A)))^2 \leq (2\dist(p, v_A) + \dist(p, \calC))^2\\
        &\leq (1+\alpha/2)\dist(p, \calC)^2 + (1+\frac{2}{\alpha})\cdot 4\dist(p, v_A)^2 \\
        &\leq (1+\alpha/2)\dist(p, \calC)^2 + (1+\frac{2}{\alpha})\cdot 4\cdot \frac{1}{\ell^2} \dist(p, \calC)^2\\
        &= (1+\alpha/2)\dist(p, \calC)^2+ (4 + \frac{8}{\alpha}) \cdot \frac{\alpha^2}{100} \cdot \dist(p, \calC)^2 \\
        &\leq (1+\alpha)\dist(p, \calC)^2.
    \end{align*}
    \end{itemize}
    Summing over all $p\in P$ and combining with the cells at level $\log(n)$, we conclude that:
    \begin{align*}
        \sum_{A \in \calA} \sum_{p \in P \cap A} \dist(p, a(A))^2 \leq (1+\alpha) \cost(P, \calC) +\frac{9}{\alpha}.
    \end{align*}
    The other direction is straightforward: since $a(A) \in \calC$, we have by definition for any $p\in A$ $\dist(p, \calC) \leq \dist(p, a(A))$, and therefore $\cost(P, \calC) \leq \sum_{A \in \calA} \sum_{p \in P \cap A} \dist(p, a(A))^2$.
\end{proof}

\section{Running Time Analysis and Handling Unknown $n$}\label{sec:extension}

\paragraph*{Running Time Analysis.}
The running time of \Cref{alg:mp}, given the values $v$, is $2^{O(d)} \cdot k \log (n)$. Therefore, in dimension $O(\log (k))$, the running time is $k^{O(1)} \log (n)$. 
Computing the values for all cell is the most expensive part of the algorithm: it is both slow and takes lot of memory, as there are  $n^{O(d)}$ many cells.

However, the algorithm \Cref{alg:mp} does not consider more than $2^{O(d)} \cdot k \log (n)$ values. In particular, it does never use the value of a cell whose parents have $\nval$ less than $\thresh$ -- where the parents of a net point $z$ are all net point at previous level with $z$ in their 1-neighborhood.
Therefore, it is not necessary to maintain any information for those, which leaves simply $2^{O(d)} \cdot k \log (n)$ many counters -- or $k^{O(1)} \log (n)$ for the high-dimensional case.

The only detail to deal with is to add new counters in memory, when the count of net point $z$ goes from below $\thresh$ to above: then, one needs to incorporate the counts of its children. For this, it is merely enough to start a new counter from scratch, and add to it all points currently present in the dataset (at most $n$). 
For two neighboring datasets, those sets of points are neighboring as well, and therefore the privacy proof follows.

\paragraph*{Knowledge of $n$.}
Knowing $n$ (the maximum number of points at any given step) is necessary 
to apply \Cref{lem:dpCountNet} with the right privacy parameters -- as we scale $\eps$ by $\log (n)$.

If $n$ is not known, we can proceed exactly as presented for the time $T$, by guessing a value of $n$ (initially some constant, say 100), and restarting the  algorithm from scratch every time the maximum number of points doubles, and scaling the privacy budget by $1/i^{1+\kappa}$ for the $i$-th algorithm.
As in the proof of \Cref{lem:boostProba}, one can show that this increases the additive error only by a factor $O\lpar \frac{\log(n)^{1+\kappa}}{\kappa}\rpar$.

\section{Extension to $(\eps, \delta)$-Differential Privacy}
\label{sec:weakPrivacy}

Instead of using $\eps$-differentially private algorithm, one could improve slightly the additive error by tolerating an $(\eps, \delta)$-differentially private one. 
The definition is the following:
For an $\eps, \delta \in \R^+$, a mechanism $\calM$ is $(\eps, \delta)$-differentially private under continual observation if, for any pair of neighboring streams $\sigma, \sigma'$, and any possible set of outcomes 
$\calS \subseteq \calX$,  it holds that 
$$\Pr\left[ \calM(\sigma) \in \calS\right]  \leq \exp(\eps) \Pr\left[ \calM(\sigma') \in \calS\right] + \delta.$$

The first improvement we can get using $(\eps, \delta)$-differential privacy is due to the advanced composition theorem of \cite{KairouzOV17}: for any $\delta > 0$, composing $k$ algorithms that are $(\eps_1, \delta_1), ..., (\eps_k, \delta_k)$-DP results in an $(\eps', \delta'$-DP algorithm, with $\eps = \sqrt{2 \log(1/\delta) \sum \eps_i^2} +  \sum \eps_i \frac{\exp(\eps_i)-1}{\exp(\eps_i)+1}$, and $\delta' = \delta + \sum \delta_i$. 
Note that, when all $\eps_i \leq 1$, $\exp(\eps_i)-1 \leq 2\eps_i$, and therefore $\eps' \leq  \sqrt{2 \log(1/\delta) \sum \eps_i} + 2\sum \eps_i^2$ 

Therefore, when running the different copies of the algorithm in \Cref{alg:boostProba}, with parameters $\eps_i = \frac{\eps}{\sqrt{s(\kappa) i^{1+\kappa} \cdot 8\log(1/\delta)}}, \delta_i = \frac{\delta}{s(\kappa) i^{1+\kappa}}$ one get an algorithm that is $(\eps', \delta')$-DP, with
\begin{align*}
    \eps' &= \sqrt{2 \log(1/\delta) \sum_{i=1}^{\log T} \eps_i^2} +  2\sum \eps_i^2\\
    &\leq \eps/2 + \eps / 2 \leq \eps\\
    \delta' &= \sum_{i=0}^{\log T} \frac{\delta}{s(\kappa) i^{1+\kappa}}\\
    &\leq \delta.
\end{align*}

Therefore, the algorithm is $(\eps, \delta)$-DP; and most importantly the quality loss due to the rescaling of $\eps$ is now essentially $\sqrt{\log(T)^{1+\kappa}}$, instead of $\log(T)^{1+\kappa}$.

In addition, we can use an $(\eps, \delta)$-DP algorithm to maintain histogram  with a better dependency in $\log T$ as well: 
\begin{lemma}[Corollary 2 in \cite{monikaHisto}]
    Let $\calU$ be a universe of size $u$, and $b$ be a given integer. Consider a stream of $T$ vectors $\sigma(t) \in [-1, 1]^u$ such that at most $b$ entries of $\sigma(t)$ are non-zero, and such that for all time $t$ all entries of $\sum_{i=1}^t \sigma(i)$ are non-negative. Then there is an $(\eps, \delta)$-DP algorithm which output a vector $h_t \in \R^u$ for each round $t$, such that, with probability $2/3$, it holds that 
    \[\forall t, \logor h_t - \sum_{i=1}^t \sigma(i)\rnor_{\infty} \leq \frac{c\sqrt{\log(1/\delta)}}{\eps} \log(T)\sqrt{b \log(6u T)},\]
    where $c$ is some constant.
\end{lemma}

Compared to  \Cref{lem:dpCountNet}, this works directly for all time steps simultaneously. Therefore, the proof of \Cref{lem:costEst} is now almost immediate, as it avoids the union-bound over $T$ steps, and using the lemma above allows to save another factor $\sqrt{\log T}$.

Combined, those two improvements yield an algorithm with additive error $k^{O(1)} \frac{\log(n)^4 \log(T)^{2+\kappa} \log(1/\delta)}{\eps}$.

\end{document}